\documentclass[aps,pra,twocolumn,showpacs,amsmath,amssymb,nofootinbib,longbibliography,superscriptaddress
]{revtex4-1}

\usepackage{algorithm}
\usepackage{algpseudocode}

\usepackage[pdftex]{graphicx}
\usepackage[usenames,dvipsnames]{xcolor}
\usepackage{epstopdf}
\usepackage[hidelinks]{hyperref}

\usepackage{orcidlink}

\usepackage[normalem]{ulem}

\graphicspath{{Figures}}

 \usepackage[utf8]{inputenc}

\usepackage{beramono}
\usepackage{listings}

\usepackage{amssymb,amsmath,amsthm}

\usepackage{algorithm}
\usepackage{algpseudocode}

\usepackage{caption}
\usepackage{subcaption}
\usepackage[T1]{fontenc}

\algrenewcommand\algorithmicrequire{\textbf{Input:}}
\algrenewcommand\algorithmicensure{\textbf{Output:}}

\newtheorem{theorem}{Theorem}
\theoremstyle{definition}
\newtheorem{definition}{Definition}

\theoremstyle{remark}

\lstdefinelanguage{Julia}%
  {morekeywords={abstract,break,case,catch,const,continue,do,else,elseif,%
      end,export,false,for,function,immutable,import,importall,if,in,%
      macro,module,otherwise,quote,return,switch,true,try,type,typealias,%
      using,while},%
   sensitive=true,%
   alsoother={\$},%
   morecomment=[l]\#,%
   morecomment=[n]{\#=}{=\#},%
   morestring=[s]{"}{"},%
   morestring=[m]{'}{'},%
     literate={é}{{\'e}}1
           {è}{{\`e}}1
           {ù}{{\`u}}1
}[keywords,comments,strings]%

\usepackage{chngcntr}

\begin{document}

\title{Part of the Strassen algorithm can speed up matrix multiplication in a parallel pebbling game
}
\author{Aaron Dayton\orcidlink{0009-0003-7573-6143}}
\email{Please direct correspondence to: aarondayton@uvic.ca}
\affiliation{Department of Physics \& Astronomy, University of Victoria, Victoria, British Columbia V8P 5C2, Canada}
\author{Thomas E.~Baker\orcidlink{0000-0002-3142-0767}}
\affiliation{Department of Physics \& Astronomy, University of Victoria, Victoria, British Columbia V8P 5C2, Canada}
\affiliation{Department of Chemistry, University of Victoria, Victoria, British Columbia V8P 5C2, Canada}

\begin{abstract}

We present a novel variant of the Strassen algorithm called the Partial Strassen algorithm, which uses a fraction of the Strassen steps to perform matrix multiplication in less time than traditional implementations. The memory footprint required to implement this algorithm is provably small whether used with a single thread or in a multi-threaded context. Data comparing the Partial Strassen algorithm at depths one, two, and three with BLAS matrix multiplication show that the three-level Partial Strassen algorithm is able to perform matrix multiplication in $80\%$ the time of BLAS with $2.25$ times the memory requirement, with a theoretical improvement of $75\%$ for sufficiently large matrices. The Partial Strassen algorithm is also compared against an implementation of the traditional Strassen algorithm, and is shown to be more performant with increasing matrix size due to decreased memory consumption. An open source implementation of the Partial Strassen algorithm for arbitrary depth and rectangular matrix multiplication is also included. 

\end{abstract}
\maketitle


\section{Introduction}

The Strassen algorithm, published in 1969, is a scheme for matrix multiplication which varies from the traditional approach implemented in basic linear algebra subprograms (BLAS)~\cite{Strassen}. It has been the subject of much interest over the past 50 years, leading to the development of a plethora of matrix multiplication algorithms which all have one common objective: to outperform BLAS matrix multiplication temporally~\cite{Winograd, Coppersmith, Karstadt, Schwartz}.

To understand the algorithm, we must view matrix multiplication
from the perspective of subdividing the two input matrices and the resulting output matrix each into 4 quartered blocks. 
\begin{equation}
    \mathbf{A} \mathbf{B} = \mathbf{C}
\end{equation}
\begin{equation}
\left(\begin{array}{c|c}
    \mathbf{A}_{1,1} & \mathbf{A}_{1,2} \\
    \hline
    \mathbf{A}_{2,1} & \mathbf{A}_{2,2}
\end{array}\right)\left(\begin{array}{c|c}
    \mathbf{B}_{1,1} & \mathbf{B}_{1,2} \\
    \hline
    \mathbf{B}_{2,1} & \mathbf{B}_{2,2}
\end{array}\right) =
\left(\begin{array}{c|c}
    \mathbf{C}_{1,1} & \mathbf{C}_{1,2} \\
    \hline
    \mathbf{C}_{2,1} & \mathbf{C}_{2,2} 
\end{array}\right)
\end{equation}
From this perspective, traditional matrix multiplication would form the resultant matrix by applying the following equations:
\begin{align}\label{Ts}
    \mathbf{T}_1 =& \;\mathbf{A}_{1,1} \mathbf{B}_{1,1} \;\;\;\;\;\;\;\; \mathbf{T}_5 = \;\mathbf{A}_{2,1} \mathbf{B}_{1,1}\nonumber \\
    \mathbf{T}_2 =& \;\mathbf{A}_{1,2} \mathbf{B}_{2,1} \;\;\;\;\;\;\;\; \mathbf{T}_6 = \;\mathbf{A}_{2,2} \mathbf{B}_{2,1}\nonumber \\
    \mathbf{T}_3 =& \;\mathbf{A}_{1,1} \mathbf{B}_{1,2} \;\;\;\;\;\;\;\; \mathbf{T}_7 = \;\mathbf{A}_{2,1} \mathbf{B}_{1,2}\nonumber \\
    \mathbf{T}_4 =& \;\mathbf{A}_{1,2} \mathbf{B}_{2,2} \;\;\;\;\;\;\;\; \mathbf{T}_8 = \;\mathbf{A}_{2,2} \mathbf{B}_{2,2}
\end{align}
where these submatrices are then combined to form $\mathbf{C}$ as
\begin{align}
    \mathbf{C} = \left(\begin{array}{c|c}
        \mathbf{C}_{1,1} & \mathbf{C}_{1,2} \\
        \hline
        \mathbf{C}_{2,1} & \mathbf{C}_{2,2}
    \end{array}\right)
    = \left(\begin{array}{c|c}
        \mathbf{T}_1 + \mathbf{T}_2 & \mathbf{T}_3 + \mathbf{T}_4 \\
        \hline
        \mathbf{T}_5 + \mathbf{T}_6 & \mathbf{T}_7 + \mathbf{T}_8
    \end{array}\right)
\end{align}
such that one can see this process involves 8 submatrix multiplications and 4 submatrix additions.

The core concept in the Strassen algorithm is that a different set of equations may be applied to obtain the same result. Those equations are as follows~\cite{Strassen}:
\begin{align}\label{Ms}
    \mathbf{M}_1 =& \;(\mathbf{A}_{1,1} + \mathbf{A}_{2,2})\cdot(\mathbf{B}_{1,1} + \mathbf{B}_{2,2})\nonumber \\
    \mathbf{M}_2 =& \;(\mathbf{A}_{2,1} + \mathbf{A}_{2,2})\cdot\mathbf{B}_{1,1}\nonumber \\
    \mathbf{M}_3 =& \;\mathbf{A}_{1,1}\cdot(\mathbf{B}_{1,2} - \mathbf{B}_{2,2})\nonumber \\
    \mathbf{M}_4 =& \;\mathbf{A}_{2,2}\cdot(\mathbf{B}_{2,1} - \mathbf{B}_{1,1})\nonumber \\
    \mathbf{M}_5 =& \;(\mathbf{A}_{1,1} + \mathbf{A}_{1,2})\cdot\mathbf{B}_{2,2}\nonumber \\
    \mathbf{M}_6 =& \;(\mathbf{A}_{2,1} - \mathbf{A}_{1,1})\cdot(\mathbf{B}_{1,1} + \mathbf{B}_{1,2})\nonumber \\
    \mathbf{M}_7 =& \;(\mathbf{A}_{1,2} - \mathbf{A}_{2,2})\cdot(\mathbf{B}_{2,1} + \mathbf{B}_{2,2})
\end{align}
where these submatrices are then combined to form $\mathbf{C}$ as
\begin{equation}\label{C}
\mathbf{C} = \left(\begin{array}{cc}
    \mathbf{M}_1 + \mathbf{M}_4 - \mathbf{M}_5 + \mathbf{M}_7 & \mathbf{M}_3 + \mathbf{M}_5 \\
    \mathbf{M}_2 + \mathbf{M}_4 & \mathbf{M}_1 - \mathbf{M}_2 + \mathbf{M}_3 + \mathbf{M}_6
\end{array}\right)
\end{equation}
which requires only 7 submatrix multiplications but 18 submatrix additions. Since matrix multiplications run in $O(n^3)$ while additions run in $O(n^2)$, where $n$ is maximum matrix dimension, we should expect to obtain computational benefits for sufficiently large matrices. This is in fact the case as found by~\cite{Huang-Smith}.

Of course, one does not necessarily need to stop here. The Strassen equations can be applied again to the 7 submatrix multiplications involved, yielding the two-leveled Strassen algorithm~\cite{Strassen}. On can then imagine recursively applying the equations until matrices could not be subdivided any further, giving the recursive Strassen algorithm which runs in $O(n^{\log_2 7}) \approx O(n^{2.81})$~\cite{Strassen}.

However, a key downside of using big-O notation becomes evident if one tries to implement a recursive Strassen algorithm: constant factors and subleading-terms are hidden. For the sake of theory, recursive Strassen is a beautiful display of the power of linear algebra and algorithmic design; for the sake of application, it is less appealing. It turns out that any implementation of the recursive Strassen algorithm has a crossover point for $n$ so large that it is realistically impractical on any machine~\cite{Huang-Smith}. In a naive implementation, one will also find that it consumes an enormous amount of memory, so even when applied to a multiplication of sufficiently large matrices, the machine in use would likely overflow the cache size or run out of memory entirely. In addition, due to the large memory burden additional computation time is required for communication of stored values within the machine. 

Many variants of the Strassen algorithm have been developed to overcome these issues~\cite{Winograd, Coppersmith, Karstadt}. Limited depth Strassen algorithms are the most applicable when it comes to a strict Strassen algorithm~\cite{Huang-Smith}. For example, a particularly good algorithm which overcomes these issues is that by Schwartz and Vaknin which uses a pebbling game and alternative basis with the aim of greatly reducing the overhead from additions and communication time~\cite{Schwartz}. A limitation of the Schwartz-Vaknin variant, however, may be its applicability to massive parallelization since it relies on a sequential set of steps wherein many threads can be applied to individual steps to obtain enhanced performance. These sequential steps are performed at the bottom level many times, so in theory multiple sets of these sequential steps could be run in parallel, however two additional temporary matrices are required for each step run in parallel making memory consumption (and potentially communication time) grow with the number of matrix operations run in parallel. Whether the Schwartz-Vaknin variant is competitive when applied with massive parallelization is yet to be determined, though it is clear that advanced techniques are required to strike a balance between the set of steps which must be run sequentially and those which may be run in parallel. Our goal is to develop new techniques towards this end.

In this paper, we will introduce a novel variant of the Strassen algorithm: the Partial Strassen algorithm. This new variant is designed with the goal of reducing memory management as the first priority, the idea being that reducing memory consumption would result in less communication time between memory addresses. This focus has led to an $O(n^3)$ algorithm which is advantageous to traditional matrix multiplication for sufficiently large matrices due to a smaller constant factor in the leading term. A highlight of the algorithm is that the total memory footprint does not increase with depth (the number of times the Strassen equations are applied to a submatrix multiplication) even when multi-threading is applied. We present a simplified description of the algorithm for context, followed by the total memory footprint required, then the total number of floating point operations (flops) for multiplication and addition operations in Sec.~\ref{MotivatedPartialStrassen}. Then the full formal description of the algorithm in Sec.~\ref{PartialStrassen}. Multi-threading is discussed in Sec.~\ref{MultiSec}.

Experimental comparisons between an implementation of the Partial Strassen algorithm, the conventional Strassen algorithm, and BLAS running on a single thread are presented and analyzed in Sec.~\ref{Results}. Finally, potential extensions of the Partial Strassen algorithm are discussed in Sec.~\ref{Extensions}.

Development of the algorithm has been done in an iterative process with different design choices. All versions created will be open sourced in the StrassOPen.jl GitHub repository~\cite{strassOPen}, with the main version extending to the multiplication of rectangular matrices. 

\section{Motivation for the Partial Strassen algorithm}\label{MotivatedPartialStrassen}

We can define a hybrid algorithm that blends both traditional matrix multiplication and the Strassen equations.\footnote{This definition may also be extended to other Strassen-like algorithms, see Sec.~\ref{StrassenLikeAlgs}}


\begin{definition}[Partial Strassen algorithm]
On application of the Strassen equations~\eqref{Ms}, choose $w<7$ of the matrices to further apply the Strassen equations to for each Strassen submatrix multiplication until some depth $N \geq 1$. The remaining $7-w$ matrices of each Strassen application have the traditional matrix multiplication applied. At the bottom level (depth $N$), traditional matrix multiplication is applied to all $7$ matrices resulting from each Strassen application at level $N-1$. Combine all resulting matrices according to Eq.~\eqref{C}.
\label{PartialStrassenDef}
\end{definition}

To perform the full matrix multiplication it is recommended to start from the finest level with the smallest matrices and coarse grain up; this is the strategy we will focus on for memory management purposes. Note that arbitrary selection of the $w$ matrices in the Partial Strassen algorithm does not necessarily lead to a good strategy for memory management. The specific set of equations from Eqs.~\eqref{Ms} is to be discussed.

The implementation we pursue here is a Partial Strassen algorithm with $w=4$ that reduces memory consumption such that it is constant with respect to depth (level). From this point on, we will simply refer to our implementation as the {\it Partial Strassen algorithm} for readability, though we recognize that other Partial Strassen variants exist in accordance with Def.~\ref{PartialStrassenDef}. A pebbling game methodology~\cite{Schwartz, Boyer, Nordstrom} can be implemented with constant size buffers which are reused to perform matrix multiplication.

\begin{definition}[Pebbling game]
Concerning the assignment of a region of memory for a part of an algorithm. When memory is being used on a given operation, a pebble can be assigned corresponding to a memory allocation required in the algorithm. When the output is generated by said operation, the corresponding pebble (memory location) may be allocated to a new operation which was dependent on the completion of the prior. Once all pebbles are placed on the final output and the corresponding operation is executed, the game is completed.
\end{definition}

\begin{figure}
    \centering
    \includegraphics[width=0.6\columnwidth]{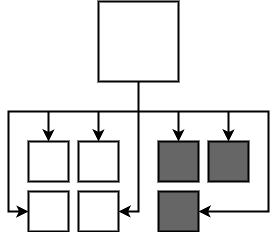}
    \caption{The memory dependence for a single application of the Strassen equations such that white nodes go into $\mathbf{C}$ and grey nodes go into the M-buffer (see text). To evaluate as a pebbling game, we start with the leaf nodes (smallest boxes) and seek to evaluate all nodes. To remove a leaf node, apply the corresponding equation from Eqs.~\eqref{Ms}. Only once all child nodes are executed may a parent node be evaluated. The resulting matrices from the leaf nodes are added according to Eq.~\eqref{C} to form $\mathbf{C}$ (the root node).
}
    \label{SingleThreadTree}
\end{figure}

\begin{figure}
    \centering
    \includegraphics[width=\columnwidth]{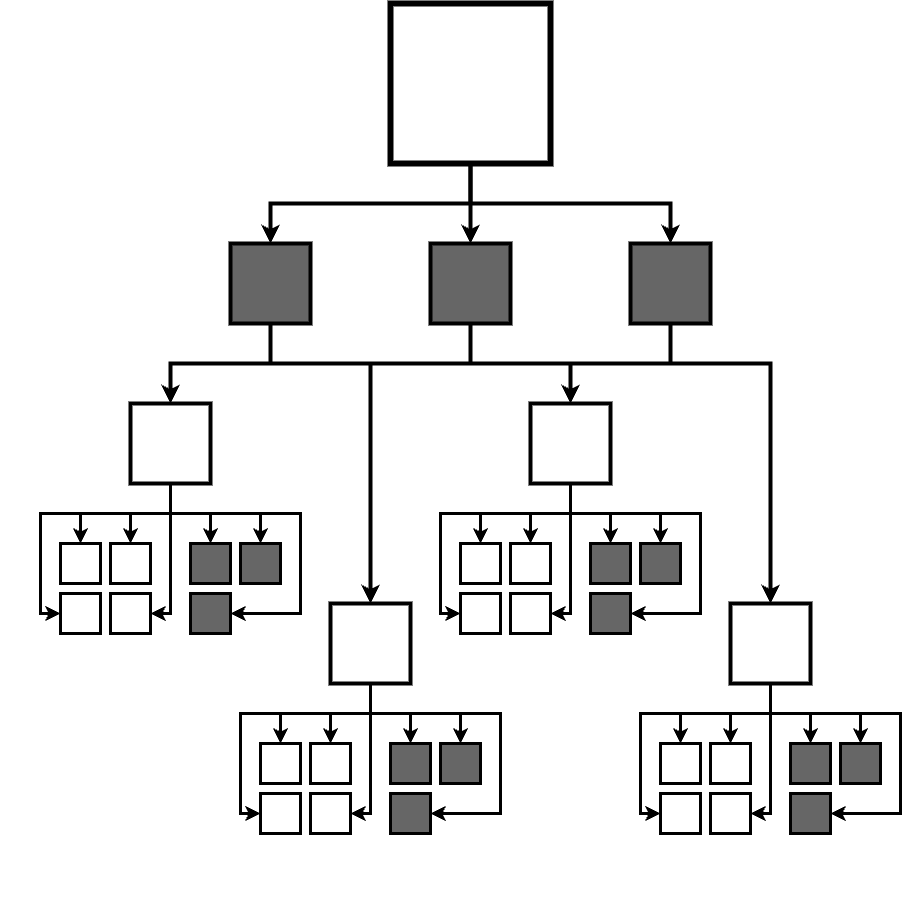}
    \caption{The memory dependence for two applications of the Strassen equations in the $w=4$ Partial Strassen scheme. The shading corresponds to the same memory organization as in Fig.~\ref{SingleThreadTree}. White nodes above the leaf nodes have the additional meaning that they were split via the Strassen equations. At the bottom level, grey nodes are stored immediately in the M-buffer without conflicts. However, at higher levels the white submatrices must be reconstructed before any grey nodes of the same depth can be computed to ensure that the M-buffer memory is not overwritten.}
    \label{SmallThreadTree}
\end{figure}

\begin{figure}
    \centering
    \includegraphics[width=\columnwidth]{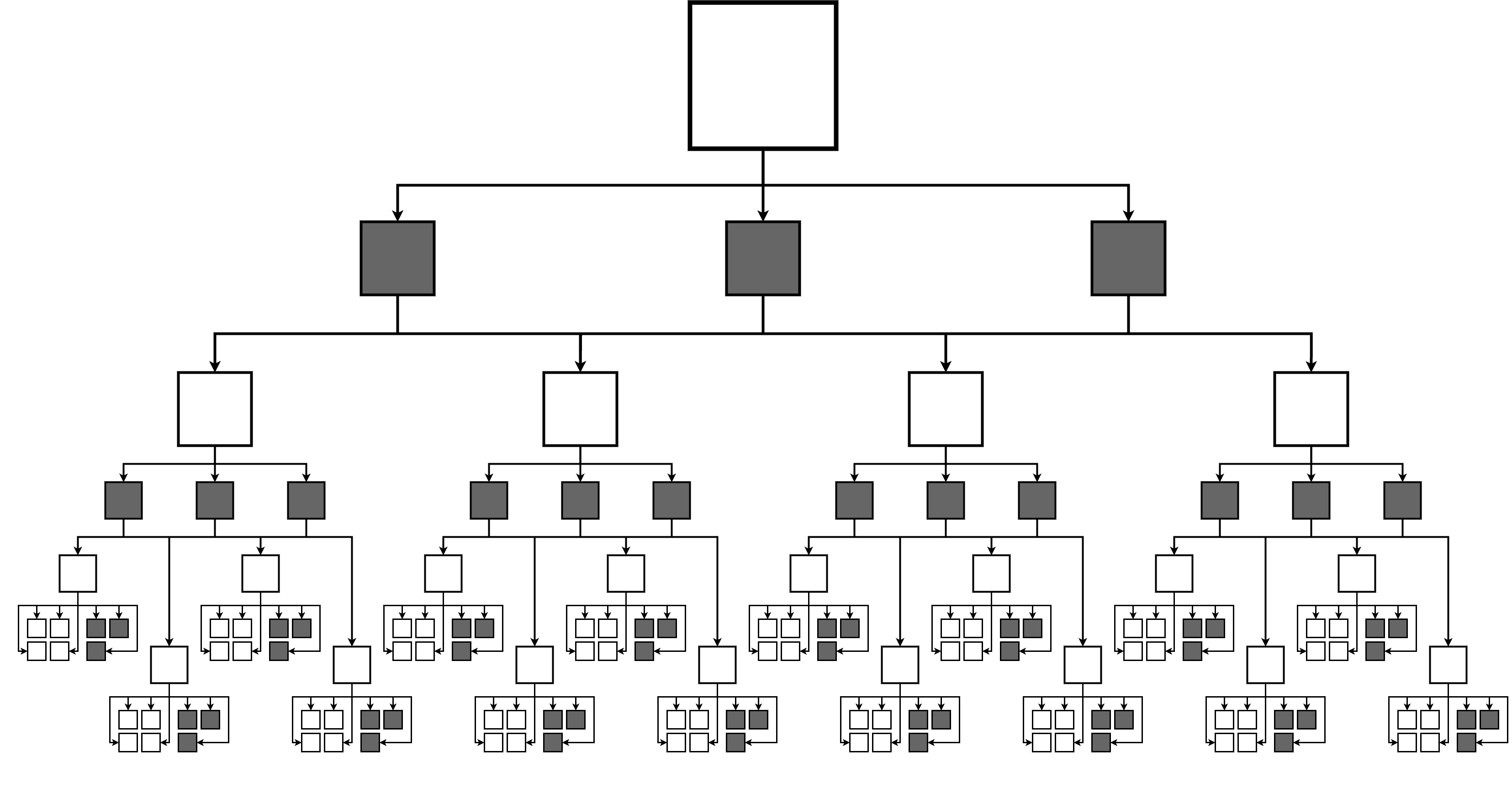}
    \caption{The memory dependence for three applications of the Strassen equations. The same rules apply as outlined in Fig.~\ref{SmallThreadTree}. All operations in a single-level can be evaluated in parallel, which is true for any depth.}
    \label{LargeThreadTree}
\end{figure}

In our implementation, the resultant matrix $\mathbf{C}$ in Eq.~\eqref{C} is used to store intermediate results throughout computation along with an additional buffer denoted $\mathbf{\mathcal{M}}$ (also called the M-buffer). This keeps the memory usage constant throughout the algorithm.

The pebbling game scheme of the Partial Strassen algorithm at depths $N=1$, 2, and 3 is depicted in Figs.~\ref{SingleThreadTree}--\ref{LargeThreadTree}, respectively, where white blocks denote matrices to be stored in $\mathbf{C}$ and grey block denote those to be stored in $\mathbf{\mathcal{M}}$. Additionally, we select white blocks to be those which are further subdivided via Eqs.~\eqref{Ms} if the maximum depth is not yet reached. Dependency indicates the order in which operations are to be performed. More details on the exact pebbling game scheme to be discussed in Sec.~\ref{MemoryManagement}. Many operations of a single-level can be performed in parallel. Hence, we extend to a {\it parallel pebbling game} in which multiple pebbles can be computed simultaneously, which will be discussed more in Sec.~\ref{MultiSec}. 

The pebbling game here utilizes a memory fractalization scheme depicted in Figs.~\ref{CFlow} and~\ref{BufferFlow}. The white blocks from Figs.~\ref{SingleThreadTree}--\ref{LargeThreadTree} are stored as depicted in Fig.~\ref{CFlow} and the grey blocks are stored in an M-buffer depicted in Fig.~\ref{BufferFlow}.

The extension of the fractal M-buffer to level $N=7$ is depicted in Fig.~\ref{HugeMBuffer}. We show the memory region in a $\Gamma$-shape for demonstration purposes. An actual implementation uses three $\frac{n^2}{4}$ sized memory regions arranged sequentially or a vectorized memory region.

\begin{figure}[H]
    \centering
    \includegraphics[width=0.8\columnwidth]{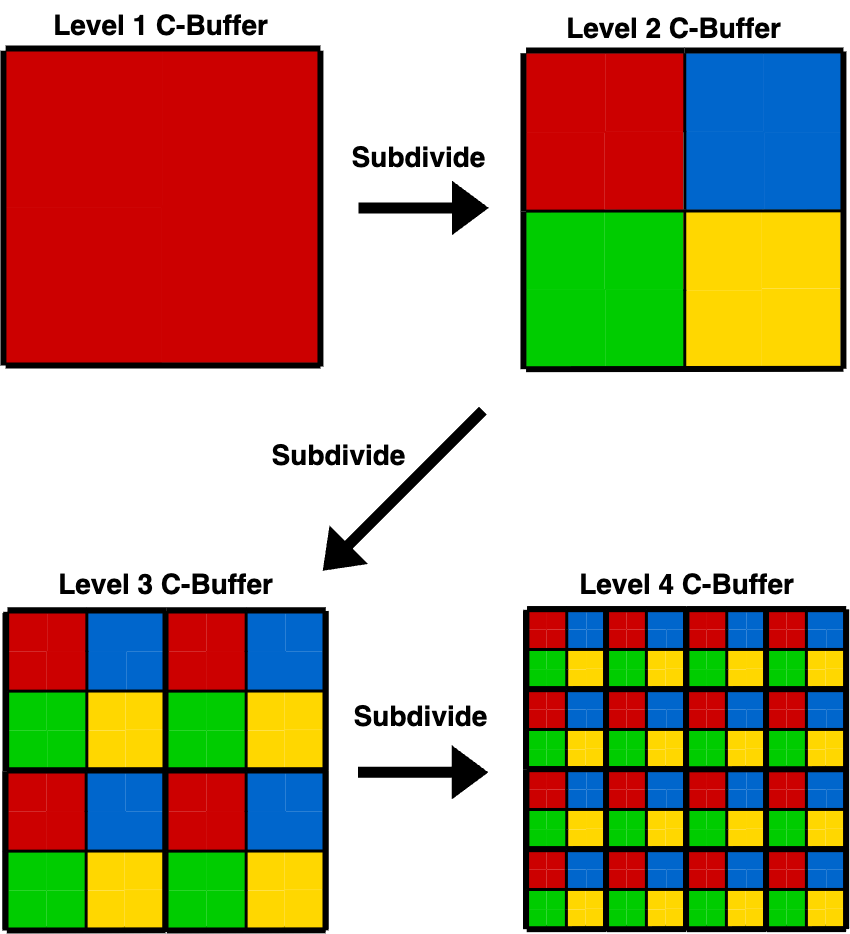}
    \caption{The C-buffer memory blocking scheme for the Partial Strassen algorithm. \textit{Level 1} contains four submatrices (all colored red) of size $\frac{n}{2} \times \frac{n}{2}$, corresponding with the white blocks for the single-level tree of Fig.~\ref{SingleThreadTree} which are to be recombined with the three red submatrices depicted in Fig.~\ref{BufferFlow} at \textit{level 1}. Each colored block at \textit{level 2} represents four submatrices of size $\frac{n}{4} \times \frac{n}{4}$ which result from an application of the Strassen equations to one of the four C-buffer submatrices of \textit{level 1} such that they are stored in the same memory region as their \textit{level 1} parent. The same pattern holds for \textit{level 3} and \textit{4}, with each applying another set of the Strassen equations to the submatrices of the previous level. This pattern holds indefinitely for all levels.
    }
    \label{CFlow}
\end{figure}

\begin{figure}[H]
    \centering
    \includegraphics[width=0.8\columnwidth]{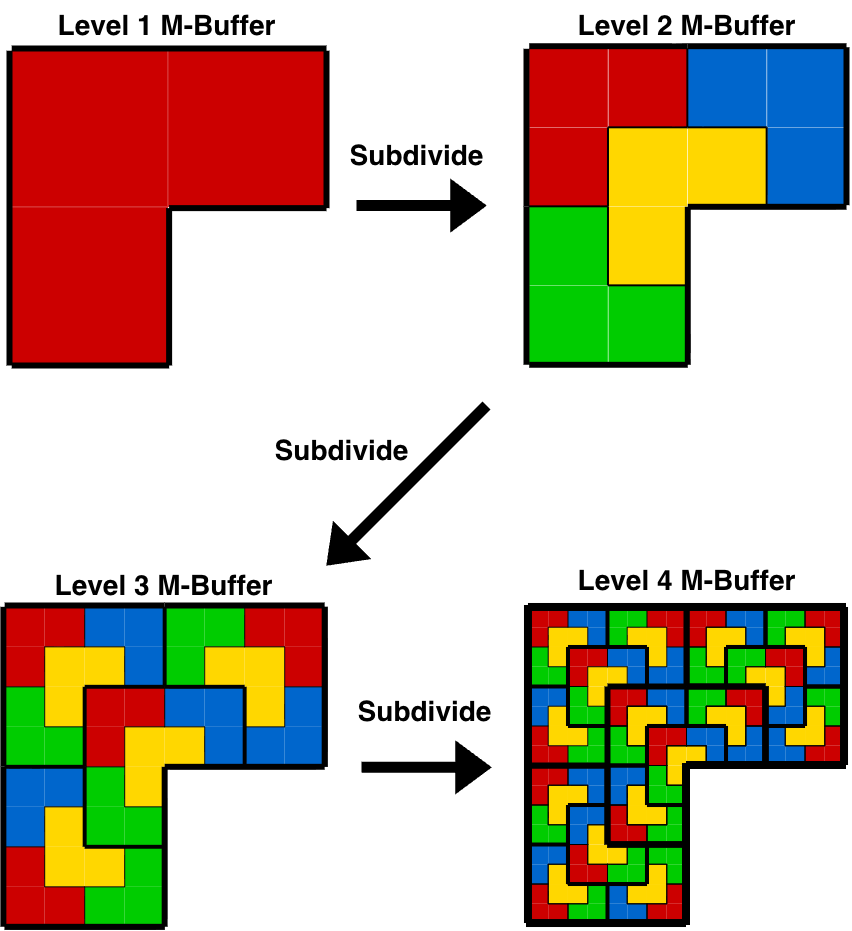}
    \caption{The M-buffer memory blocking scheme for the Partial Strassen algorithm. \textit{Level 1}  contains three submatrices (all colored red) of size $\frac{n}{2} \times \frac{n}{2}$, corresponding with the grey blocks for the single-level tree of Fig.~\ref{SingleThreadTree} which are to be recombined with the four red submatrices depicted in Fig.~\ref{CFlow} at \textit{level 1}. Each colored region at \textit{level 2} represents three submatrices of size $\frac{n}{4} \times \frac{n}{4}$ which result from an application of the Strassen equations to one of the four C-block submatrices of \textit{level 1}. The memory buffer used at \textit{level 1} gets distributed amongst all four triplets of \textit{level 2}. A similar pattern holds at \textit{levels 3} and \textit{4} (and indefinitely).
    }
    \label{BufferFlow}
\end{figure}

\begin{figure}[H]
    \centering
    \includegraphics[width=0.7\columnwidth]{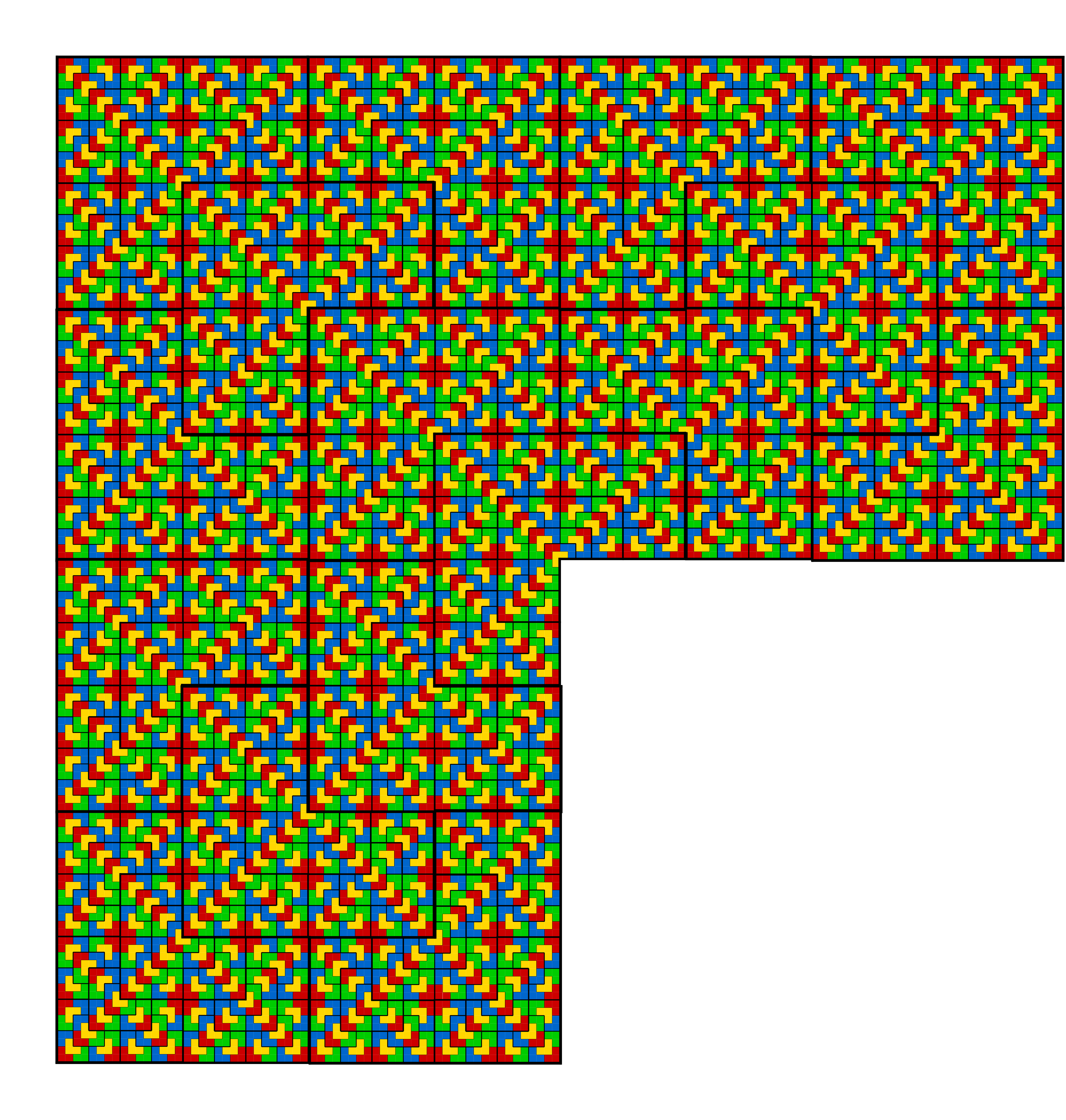}
    \caption{The fractalized M-buffer for a depth $N=7$ Partial Strassen where different colored regions indicate distinct memory regions for the storage of specific $\mathbf{M}$ submatrices. The memory in actual implementation is stored sequentially, but this illustrates the sub-divided pattern of the memory packing scheme in the algorithm.}
    \label{HugeMBuffer}
\end{figure}

The overview of the method in Alg.~\ref{simple_partial_strassen} describes the essential steps in the Partial Strassen algorithm as a broad overview for $w=4$, each of which will be expanded upon and justified in subsequent sections.

\begin{algorithm}[H]
    \caption{Simple $w=4$ Partial Strassen Algorithm}
    \label{simple_partial_strassen}
    \begin{algorithmic}[1]
        \Require Matrix $A$, matrix $B$, and depth $N$
        \State Initialize $\mathbf{C}$ of size $n^2$ \Comment{Fig.~\ref{CFlow}}
        \State Initialize $\mathbf{\mathcal{M}}$ of size $\frac{3}{4}n^2$ \Comment{Fig.~\ref{BufferFlow}}
        \State Initialize submatrix addition buffers 

        \State Allocate memory in $\mathbf{C}$ for 4 of the $\mathbf{M}_i$ \Comment{Eqs.~\eqref{Ms}} 
        \State Allocate memory in $\mathbf{\mathcal{M}}$ for the remaining $\mathbf{M}_i$ \Comment{Eqs.~\eqref{Ms}}

        \For{$1$ to $N-1$}
            \State \parbox[t]{0.8\linewidth}{Apply Eqs.~\eqref{Ms} only to matrix multiplications allocated to $\mathbf{C}$. 4 out of 7 of the submatrices produced by this division will have space allocated to their parent's region in $\mathbf{C}$ while the others have space allocated in $\mathbf{\mathcal{M}}$}
        \EndFor

        \For{$N-1$ to $1$}
            \State Compute and store leaf multiplications \Comment{Eqs.~\eqref{Ms}}
            \State Compute and store parent node additions \Comment{Eq.~\eqref{C}}
            \State Prune computed nodes from tree \Comment{Fig.~\ref{LargeThreadTree}}
        \EndFor

        \Ensure $\mathbf{C}$
    \end{algorithmic}
\end{algorithm}

The Partial Strassen algorithm can be applied to a specific depth or fully recursively (however increasing depth sees diminishing returns). The main deviation from the multi-leveled or recursive Strassen algorithm is that not all the submatrix multiplications for a given level have the Strassen equations applied (some are just computed with traditional matrix multiplication). This results in the time complexity of the algorithm scaling as $O(n^3)$. Thus, asymptotically, the scaling is just as it is with regular matrix multiplication. This algorithm is useful because the $n^3$ prefactor scales ideally as $75\%$ of the cost of regular matrix multiplication (in terms of flops) while requiring a reserved increase in memory consumption (as we will see shortly).

Notice that the single-level Partial Strassen algorithm is just a traditional single-level Strassen algorithm (with the enforced memory scheme of our implementation), though we still refer to this as a Partial Strassen algorithm for the sake of simplicity as the need to constantly distinguish between cases would be cumbersome.

\section{Partial Strassen Algorithm}\label{PartialStrassen}

With the basic ideas of how the algorithm should work, we will now rigorously prove claims on the performance and memory consumption of the algorithm for the $w=4$ algorithm. Justification will be built up for the given procedure such that the Partial Strassen algorithm is the logical endpoint.

\subsection{Memory Management}\label{MemoryManagement}

We first discuss the formal memory requirement for use of the Partial Strassen before proving time-complexity.

\begin{theorem}
The memory required for the Partial Strassen algorithm with a single thread is $2.25n^2$.
\label{memoryscaling_single}
\end{theorem}

\begin{proof}
Suppose you are given two $n \times n$ matrices to multiply together. Performing Eqs.~\eqref{Ms} on the input matrices yields 7 $\mathbf{M}_i$ matrices of size $\frac{n}{2} \times \frac{n}{2}$, requires the storage of $\frac{7}{4}n^2$ floating point numbers in memory.

One could further apply the Strassen equations for each matrix multiplication involved in Eqs.~\eqref{Ms}, as is done in the traditional two-level Strassen algorithm. This would require storing 49 $\mathbf{M}_i$ matrices of size $\frac{n}{4} \times \frac{n}{4}$ for a total of $\frac{49}{16}n^2$. In general, applying the Strassen equations $N$ times for an $N$-level Strassen requires
\begin{equation}
    \left(\frac{7}{4}\right)^N n^2
\end{equation}
memory. Thus, memory consumption grows exponentially with $N$ for the multi-leveled Strassen algorithm.

However, if one chooses to reapply the Strassen equations to only 4 out of the 7 $\mathbf{M}_i$ matrices of size $\frac{n}{2} \times \frac{n}{2}$ which result from a single application of Eqs.~\eqref{Ms}, then $4 \cdot 7 = 28$ matrices of size $\frac{n}{4} \times \frac{n}{4}$ must be stored. This requires $\frac{7}{4} n^2$ memory. In general, at depth $k$ the results from submatrix multiplications require
\begin{equation}
4^k \cdot 7 \cdot \frac{n}{2^{k+1}} \cdot \frac{n}{2^{k+1}} = \frac{7}{4} n^2
\end{equation}
memory.

Addition of $\mathbf{A}$ and $\mathbf{B}$ submatrices must be performed in buffers as well. With a single thread these can be done sequentially with no addition will being larger than $\frac{n^2}{4}$, thus an additional $\frac{n^2}{2}$ memory is required for the A- and B-buffers whose contents will be multiplied together to form M-matrices according to Eqs.~\eqref{Ms}.

Thus, $\frac{9}{4}n^2$ memory is required to run the Partial Strassen algorithm with a single thread.
\end{proof}

In what follows, the actual implementation of the pebbling game for efficient memory usage is detailed.

\subsubsection{Memory pebbling scheme}\label{PebblingScheme}

The memory required for the $w=4$ algorithm is demonstrated here.

We define the notation $\mathbf{M}_{g,h}$ to indicate the $h^{\text{th}}$ $\mathbf{M}$ matrix produced by applying Eqs.~\eqref{Ms} to the matrix multiplication of the $g^{\text{th}}$ $\mathbf{M}$ matrix produced by the first application of these equations. This notation extends to any depth as $\mathbf{M}_{g,h,\ldots,z}$ where $g,\; h, \ldots,\; z$ are all integers from $1$ to $7$. Appendix~\ref{DeriveStrassenEquation} contains an example of this labelling for multi-leveled operations.

Consider applying the Strassen equations in the outlined procedure $N$ times. At the bottom level, conventional matrix multiplication is used to compute all $\mathbf{M}_{\ldots, i}$ matrices according to Eqs.~\eqref{Ms}. The results fill up $\frac{7}{4}n^2$ memory. These matrices are then added together according to equation Eq.~\eqref{C}. The result of this recombination yields the $4^{N-1}$ submatrices of size $\frac{1}{4^{N-1}}n^2$, requiring $n^2$ total memory. The $3 \cdot 4^{N-2}$ submatrices at level $N-1$ which did not have the Strassen equations applied then must be computed with conventional matrix multiplication, requiring $\frac{3}{4}n^2$ memory. This fact holds in general, at any level.

This allows us to use memory in a pebbling game strategy such that only $\frac{7}{4}n^2$ memory is used for $\mathbf{M}_{\ldots,i}$ matrices at any time. Allocate $n^2$ memory for matrices which will be subdivided via the Strassen equations and $\frac{3}{4}n^2$ memory for those which will not be subdivided. We may use the resultant matrix $\mathbf{C}$ as the $n^2$ buffer and an external M-buffer, $\mathbf{\mathcal{M}}$, for the $\frac{3}{4}n^2$ memory.

The memory dependence structure is represented with the trees presented in Figs.~\ref{SingleThreadTree}--\ref{LargeThreadTree}. The memory is split between $\mathbf{C}$ and $\mathbf{\mathcal{M}}$ as depicted in Figs.~\ref{CFlow} and~\ref{BufferFlow}. Notice that at each level in these figures there are $4^{k-1}$ sets of seven blocks split between the two buffers (where $k$ is the level number).

This memory buffer fractalization scheme along with the constraint that the outlined sets of operations must be evaluated sequentially (where all operations within a single set may be computed in parallel) results in the Partial Strassen algorithm requiring only $\frac{7}{4}n^2$ memory for $\mathbf{M}_{\ldots,i}$ matrices.

It is important which $\mathbf{M}_i$ matrices are stored where in $\mathbf{C}$ and the M-buffer. Any $\mathbf{M}_i$ matrix stored in $\mathbf{C}$ must be placed where it belongs according to Eq.~\eqref{C}. There are multiple ways to choose how matrices are stored, but we select the following:
\begin{equation}\label{Cstore}
    \mathbf{C} \leftarrow \left(\begin{array}{cc}
        \mathbf{M}_4 & \mathbf{M}_5 \\
        \mathbf{M}_2 & \mathbf{M}_6
    \end{array}\right)
\end{equation}

\begin{equation}\label{Dstore}
    \mathbf{\mathcal{M}} \leftarrow \left(\begin{array}{cc}
        \mathbf{M}_1 & \mathbf{M}_3 \\
        \mathbf{M}_7 & 0
    \end{array}\right)
\end{equation}
where $\mathbf{\mathcal{M}}$ is the M-buffer with the zeros representing unoccupied memory. $\mathbf{\mathcal{M}}$ could be stored as a vector of the three submatrices such that this zero-valued region would not exist, but we will leave it for representative purposes.

With this selection, the two-level storage scheme then takes the form
\begin{equation}\label{level2-Cstore}
    \mathbf{C} \leftarrow \left(\begin{array}{cccc}
        \mathbf{M}_{4,4} & \mathbf{M}_{4,5} & \mathbf{M}_{5,4} & \mathbf{M}_{5,5} \\
        \mathbf{M}_{4,2} & \mathbf{M}_{4,6} & \mathbf{M}_{5,2} & \mathbf{M}_{5,6} \\
        \mathbf{M}_{2,4} & \mathbf{M}_{2,5} & \mathbf{M}_{6,4} & \mathbf{M}_{6,5} \\
        \mathbf{M}_{2,2} & \mathbf{M}_{2,6} & \mathbf{M}_{6,2} & \mathbf{M}_{6,6}
    \end{array}\right)
\end{equation}
and
\begin{equation}\label{level2-Dstore}
    \mathbf{\mathcal{M}} \leftarrow \left(\begin{array}{cccc}
        \mathbf{M}_{4,1} & \mathbf{M}_{4,3} & \mathbf{M}_{5,7} & \mathbf{M}_{5,1} \\
        \mathbf{M}_{4,7} & \mathbf{M}_{6,1} & \mathbf{M}_{6,3} & \mathbf{M}_{5,3} \\
        \mathbf{M}_{2,3} & \mathbf{M}_{6,7} & 0                & 0 \\
        \mathbf{M}_{2,1} & \mathbf{M}_{2,7} & 0                & 0
    \end{array}\right).
\end{equation}

On top of this, memory is also required for the submatrix additions involved in Eqs.~\eqref{Ms}. An A-buffer and B-buffer must store these intermediate results. On a single thread, only one of Eqs.~\eqref{Ms} may be performed at a time, so just $\frac{1}{4}n^2$ memory is required each of these addition buffers. The total memory requirement for a single-threaded version is $n^2 + \frac{3}{4} n^2 + 2\frac{1}{4} n^2 = \frac{9}{4} n^2$ floating point numbers in accordance with Thm.~\ref{memoryscaling_single}.

An alternate storage method which may have applicability to GPU implementations due to specific characteristics is discussed in Appendix.~\ref{167method}.

\subsubsection{Recombination step}

Now that the memory structure has been covered, we can describe how the $\mathbf{M}_i$ matrices are recombined via Eqs.~\eqref{C}. Since some $\mathbf{M}_i$ matrices were stored in $\mathbf{C}$, we must be careful not to overwrite them before adding them to the other locations where they must be applied. This produces a sequential set of additions in the algorithm required for transitioning between levels. Note that while the additions of any set of 7 matrices from Eqs.~\eqref{Ms} must be performed sequentially, multiple distinct sets may be computed in parallel.

Consider producing all matrices according to Eqs.~\eqref{Ms} to obtain

\begin{equation}
    \mathbf{C} = \left(\begin{array}{cc}
        \mathbf{M}_4 & \mathbf{M}_5 \\
        \mathbf{M}_2 & \mathbf{M}_6
    \end{array}\right)
\end{equation}

\begin{equation}
    \mathbf{\mathcal{M}} = \left(\begin{array}{cc}
        \mathbf{M}_1 & \mathbf{M}_3 \\
        \mathbf{M}_7 & 0
    \end{array}\right)
\end{equation}
as we saw before. Since $\mathbf{M}_6$ does not need to be added to any other submatrix, we may add to it first to obtain

\begin{equation}
    \mathbf{C} = \left(\begin{array}{cc}
        \mathbf{M}_4~~ & \mathbf{M}_5 \\
        \mathbf{M}_2~~ & \mathbf{M}_6 + \mathbf{M}_1 - \mathbf{M}_2 + \mathbf{M}_3
    \end{array}\right)
\end{equation}
so that $\mathbf{M}_2$ is no longer needed. We can then compute
\begin{equation}
    \mathbf{C} = \left(\begin{array}{cc}
        \mathbf{M}_4~~                & \mathbf{M}_5 \\
        \mathbf{M}_2 + \mathbf{M}_4~~ & \mathbf{M}_6 + \mathbf{M}_1 - \mathbf{M}_2 + \mathbf{M}_3
    \end{array}\right)
\end{equation}
and then add $\mathbf{M}_4$ as
\begin{equation}
    \mathbf{C} = \left(\begin{array}{cc}
        \mathbf{M}_4 + \mathbf{M}_1 - \mathbf{M}_5 + \mathbf{M}_7 & \mathbf{M}_5 \\
        \mathbf{M}_2 + \mathbf{M}_4                               & \mathbf{M}_6 + \mathbf{M}_1 - \mathbf{M}_2 + \mathbf{M}_3
    \end{array}\right)
\end{equation}
which can then be completed to give
\begin{equation}
    \mathbf{C} = \left(\begin{array}{cc}
        \mathbf{M}_4 + \mathbf{M}_1 - \mathbf{M}_5 + \mathbf{M}_7 & \mathbf{M}_5 + \mathbf{M}_3\\
        \mathbf{M}_2 + \mathbf{M}_4                               & \mathbf{M}_6 + \mathbf{M}_1 - \mathbf{M}_2 + \mathbf{M}_3
    \end{array}\right)
\end{equation}
which is Eq.~\eqref{C} reordered by the operations we chose in sequence.

To see an example of how a given $\mathbf{M}$ of a general depth 2 Strassen algorithm is applied order to form $\mathbf{C}$, see Appendix~\ref{EquationPropagation}.

\subsection{Counting Operations}

To count the total number of flops, each matrix computation must be considered. First, flops involved in matrix multiplications will be counted, then those involved in matrix additions will be counted.

It can then be shown that the number of flops required to run the algorithm scales as follows.

\begin{theorem}\label{singlethread_time}
The number of flops to perform the submatrix multiplications in the Partial Strassen algorithm asymptotically approaches $75\%$ of those required for traditional matrix multiplication as depth $N$ increases.
\end{theorem}

\begin{proof}

At level $k$ there will be $3 \cdot 4^{k-1}$ matrix multiplications required for the $\mathbf{M}_{\ldots,1}$, $\mathbf{M}_{\ldots,6}$, and $\mathbf{M}_{\ldots,7}$ submatrices of size $\frac{n}{2^k} \times \frac{n}{2^k}$ which will be stored in the M-buffer. This leads to the following statement

For $N$ total levels, the number of multiplication flops for external matrices is
\begin{equation}\label{ExternalMults}
    T_{\text{external}}(N) = \sum_{k=1}^{N} 3 \cdot 4^{k-1} \left( \frac{n}{2^k} \right)^3
\end{equation}

At the depth $N$, conventional matrix multiplication is applied to all $\mathbf{M}_{\ldots,i}$ submatrices. $3 \cdot 4^{N-1}$ of these multiplications have been accounted for in Eq.~\eqref{ExternalMults}. This leaves $4^{N}$ more multiplications for matrices of size $\frac{n}{2^N} \times \frac{n}{2^N}$, yielding an operation count for bottom level C-block submatrices as 

\begin{equation}\label{BottomMults}
    T_{\text{Cblock}}(N) = 4^N \left( \frac{n}{2^N} \right)^3 
\end{equation}

Notice that these multiplications for submatrices stored in $\mathbf{C}$ only happen at the bottom level since all submatrices stored in $\mathbf{C}$ at higher levels are formed via addition of submatrices from the preceding level and those stored in the M-buffer.

The total number of multiplication flops is then
\begin{align}\label{MultiplicationFlops}
    T_{\text{mult}}(N) &= T_{\text{external}}(N) + T_{\text{Cblock}}(N) \nonumber \\
      &= \left[ \sum_{k=1}^{N} 3 \cdot 4^{k-1} \left( \frac{n}{2^k} \right)^3 \right] + 4^N \left( \frac{n}{2^N} \right)^3 \nonumber \\
      &= \left[ 7 \sum_{k=1}^{N} 4^{k-1} \left( \frac{n}{2^k} \right)^3 \right] - \left[ 4 \sum_{k=1}^{N-1} 4^{k-1} \left( \frac{n}{2^k} \right)^3 \right] \nonumber \\
      &= \left[ 7 \sum_{k=1}^{N} \frac{4^{k-1}}{4^k} \frac{n^3}{2^k} \right] - \left[ 4 \sum_{k=1}^{N-1} \frac{4^{k-1}}{4^k} \frac{n^3}{2^k} \right] \nonumber \\
      &= \left(\frac{7}{8} \left[ \sum_{k=0}^{N-1} \frac{1}{2^k} \right] - \frac{1}{2} \left[ \sum_{k=0}^{N-2} \frac{1}{2^k} \right]\right)n^3 \nonumber \\
      &= \left(\frac{7}{4} \left( 1 - \left(\frac{1}{2}\right)^{N} \right) - \left( 1 - \left(\frac{1}{2}\right)^{N-1} \right) \right)n^3 \nonumber \\
      &= \left(\frac{3}{4} - \frac{7}{4}\left(\frac{1}{2}\right)^{N} + \left(\frac{1}{2}\right)^{N-1} \right)n^3 \nonumber \\
      &= \left(\frac{3}{4} + \frac{1}{8}\left(\frac{1}{2}\right)^{N-1} \right)n^3
\end{align}
For $N=1$ we obtain $T_{\text{mult}}(1) = 0.875 n^3$ in agreement with the single-level Strassen algorithm~\cite{Strassen}, and for sufficiently large $N$ this converges to $0.75 n^3$.
\end{proof}

There is also a sub-leading cost to constructing the $\mathbf{M}_{\ldots,i}$ matrices from the additions involved in Eqs.~\eqref{Ms}. The number of addition flops incurred for this step follows a complicated pattern and so are given up to depth $N=9$ in Table~\ref{AdditionOps} as obtained by computational means (the algorithm was implemented and the number of addition operations were counted directly). This table discounts the initial storage of a matrix in any buffer region as these are technically not floating point operations, but rather communication operations. Only additions are counted.

\begin{table}
    \begin{center}
        \begin{tabular}{|| l | l ||}
            \hline
            \textit{\textbf{Depth N}} & \textit{\textbf{Addition operations of Eqs.~\eqref{Ms}}} \\
            \hline\hline
            $1$ & $2.5 n^2$ \\
            $2$ & $7.5 n^2$ \\
            $3$ & $16.625 n^2$ \\
            $4$ & $32.594 n^2$ \\
            $5$ & $59.977 n^2$ \\
            $6$ & $106.490 n^2$ \\
            $7$ & $185.217 n^2$ \\
            $8$ & $318.419 n^2$ \\
            $9$ & $544.105 n^2$ \\
            \hline
        \end{tabular}
    \end{center}
    \caption{Theoretical number of addition flops required for level 1-9 of the Partial Strassen Algorithm from computation of Eqs.~\eqref{Ms} rounded to 3 significant figures.}
    \label{AdditionOps}
\end{table}

Part of the total addition operations for a given level which can be counted explicitly in a compact formula are those of the C-buffer submatrix construction.

\begin{theorem}
The total number of matrix additions for matrix recombination of the C-buffer in the Partial Strassen algorithm is $2 N n^2$.
\end{theorem}

\begin{proof}
At the $k^{\text{th}}$ level $7 \cdot 4^{k-1}$ submatrices must be combined via Eq.~\eqref{C} involving 8 submatrix additions of size $\frac{n}{2^k} \times \frac{n}{2^k}$. So for $N$ levels these additions cost the following number of flops:
\begin{align}
    T_{\text{construction}}& = 8 \sum_{k=1}^N 4^{k-1} \left( \frac{n}{2^k} \right)^2 \nonumber \\
        &= 2 \sum_{k=1}^N n^2 \nonumber \\
        &= 2 N n^2
\end{align}
\end{proof}

The total operation count for up to depth $N=9$ is given in Table~\ref{AllOps}. Notice that there are diminishing returns on the factor associated with the $n^3$ term with increasing depth.

\begin{table}
    \begin{center}
        \begin{tabular}{|| l | l ||}
            \hline
            \textit{\textbf{Depth N}} & \textit{\textbf{Total floating point operations}} \\
            \hline\hline
            $1$ & $0.875 n^3 + 4.5 n^2$ \\
            $2$ & $0.813 n^3 + 11.5 n^2$ \\
            $3$ & $0.781 n^3 + 22.625 n^2$ \\
            $4$ & $0.766 n^3 + 40.594 n^2$ \\
            $5$ & $0.758 n^3 + 69.977 n^2$ \\
            $6$ & $0.754 n^3 + 118.490 n^2$ \\
            $7$ & $0.752 n^3 + 199.217 n^2$ \\
            $8$ & $0.751 n^3 + 334.419 n^2$ \\
            $9$ & $0.750 n^3 + 562.105 n^2$ \\
            \hline
        \end{tabular}
    \end{center}
    \caption{Theoretical total number of flops required for level 1-9 of the Partial Strassen Algorithm rounded to 3 significant figures.}
    \label{AllOps}
\end{table}

\subsection{Definition of the Partial Strassen Algorithm}\label{FullPartialStrassen}

Now that we have worked out the technical details involved in the algorithm and proven several limiting behaviours, we may present the algorithm in more detail. There are two algorithms which we will present here. First, how to obtain the relevant equations and memory locations in Alg.~\ref{partial_strassen_eqs} (this only ever needs to be computed once), then how to compute matrix multiplication via those equations in Alg.~\ref{partial_strassen}

\begin{algorithm}[H]
    \caption{Partial Strassen Equation Derivation Procedure}
    \label{partial_strassen_eqs}
    \begin{algorithmic}[1]
        \State Apply Eqs.~\eqref{Ms} to a matrix multiplication. Allocate $\mathbf{M}_2$, $\mathbf{M}_4$, $\mathbf{M}_5$, and $\mathbf{M}_6$ to $\mathbf{C}$, and allocate $\mathbf{M}_1$, $\mathbf{M}_3$, and $\mathbf{M}_7$ to the M-buffer as outlined in Sec.~\ref{MemoryManagement}. Blocks in $\mathbf{C}$ are allocated such that they align with the required location in Eq.~\eqref{C}. Do not actually perform any matrix computations, just designate the space for future use.

        \State Apply Eqs.~\eqref{Ms} to each matrix multiplication involved in $\mathbf{M}_2$, $\mathbf{M}_4$, $\mathbf{M}_5$, and $\mathbf{M}_6$. Allocate space in each parent block of $\mathbf{C}$ for $\mathbf{M}_{a, 2}$, $\mathbf{M}_{a, 4}$, $\mathbf{M}_{a, 5}$, and $\mathbf{M}_{a, 6}$, $\forall a \in \{2,4,5,6\}$; and allocate space in the M-buffer for $\mathbf{M}_{a, 1}$, $\mathbf{M}_{a, 3}$, and $\mathbf{M}_{a, 7}$ according to the memory management scheme outlined in Sec.~\ref{MemoryManagement}. Only keep track of which blocks are to be allocated where.

        \State Repeat the process of Step (4) on each $\mathbf{M}_{a, \ldots, 2}$, $\mathbf{M}_{a, \ldots, 4}$, $\mathbf{M}_{a, \ldots, 5}$, $\mathbf{M}_{a, \ldots, 6}$ component matrix until the desired depth of the algorithm is reached.

        \State Return the resulting list of operations and their associated memory locations.
    \end{algorithmic}
\end{algorithm}

\begin{algorithm}[H]
    \caption{The Partial Strassen Algorithm}
    \label{partial_strassen}
    \begin{algorithmic}[1]
        \State Initialize external fixed size buffers for submatrix additions of $\mathbf{A}$ and $\mathbf{B}$ (the A- and B-buffers) with each being of size $\frac{n}{2} \times \frac{n}{2}$ for single-threading (for more details on the multi-threading case see Sec.~\ref{MemoryMulti}).

        \State Initialize the resultant matrix $\mathbf{C}$ of size $n^2$ and the external fixed size buffer $\mathbf{\mathcal{M}}$ of size $\frac{3}{4}n^2$ for storing intermediate results.

        \State Using the information obtained from Alg.~\ref{partial_strassen_eqs}, compute the bottom level additions and multiplications to construct each $\mathbf{M}_{a, \ldots, x, y, z}$, ($a,\ldots,x,y \in \{2,4,5,6\}$; $z \in \{1,\ldots,7\}$) and store their result in the pre-allocated location.

        \State Apply Eq.~\eqref{C} to each group for 7 $\mathbf{M}_{a, \ldots, x, y, z}$ matrices with matching $a, \ldots, x, y$ and store the result to the pre-allocated locations in $\mathbf{C}$ to obtain $\mathbf{M}_{a, \ldots, x, y}$, $y \in \{2,4,5,6\}$.

        \State Perform the additions and matrix multiplications to obtain $\mathbf{M}_{a, \ldots, x, \gamma}$ where $\gamma \in \{1,3,7\}$ and store the results in the $\mathbf{\mathcal{M}}$.

        \State Repeat the process outlined in Steps (4) and (5) until all matrices are recombined to form the resultant matrix $\mathbf{C}$.

        \State Return $\mathbf{C}$.
    \end{algorithmic}
\end{algorithm}

To see an example of how to derive one of the second-level equations, an example of the recombination step, and the full set of $\mathbf{M}_{i}$, $\mathbf{M}_{j,k}$ equations of the two-level Partial Strassen algorithm see Appendix.~\ref{StrassenEquationAppendix}.

\section{Multi-threading and the Partial Strassen Algorithm}\label{MultiSec}

We have so far commented on the single-threaded version of the algorithm. The possibility for a multi-threaded version was also briefly mentioned. Here we investigate that possibility and find that the requirements for the algorithm are largely unchanged from the single-threaded version. This implies the possibility for a practical multi-threading implementation and reduction in some prefactors for such a version of the algorithm.

\subsection{Memory Requirements for Multi-Threading}\label{MemoryMulti}

In the single-threaded implementation, two buffers of size $\frac{n^2}{4}$ are required for the addition of $\mathbf{A}$ and $\mathbf{B}$ submatrices. In the multi-threading approach, multiple submatrix additions must happen in parallel. The question then becomes, how much memory is required to do this given some number of threads? It turns out, there is a ceiling on the memory required for these submatrix additions.

\begin{theorem}
A multi-threaded version of the Partial Strassen algorithm requires at most $5.25n^2$ memory if evaluating all leaf nodes at once.
\label{worstcase_threading}
\end{theorem}

\begin{proof}
At the bottom level $N$, if every single matrix multiplication is run in parallel, then the total amount of memory required for the submatrix additions involved in Eq.~\eqref{Ms} is
\begin{equation}
    2 \cdot 7 \cdot 4^{N-1} \frac{n^2}{4^N} = \frac{7}{2}n^2
\end{equation}
where the factor of $2$ is for the two addition buffers, $7 \cdot 4^{N-1}$ comes from the total number of blocks, and $\frac{n^2}{4^N}$ is the matrix size. For each level above, one would require memory as
\begin{equation}
    2 \cdot 3 \cdot 4^{k-1} \frac{n^2}{4^k} = \frac{3}{2}n^2
\end{equation}
where $k < N$. Each level must be computed sequentially and the memory of the lowest level can be reused for higher levels. 

So the upper bound on the total memory required for submatrix additions is $\frac{7}{2}n^2$ such that the total upper bound on memory is $n^2 + \frac{3}{4}n^2 + \frac{7}{2}n^2 = \frac{21}{4}n^2$ (though $2n^2$ of this memory would sit idle for the majority of the algorithm's runtime).
\end{proof}

This establishes an upper bound on the memory required by any multithreaded Partial Strassen algorithm. The importance is that this is not dependent on the depth chosen, just as in the single-threaded case. This does not represent the optimal scaling of the memory, however.

A choice with reduced memory consumption is to apply
\begin{equation}
\tau \equiv 3 \cdot 4^{N-1}
\end{equation}
threads (or some integer multiple of $\tau$ such that multiple threads are applied to each submatrix multiplication) such that only $\tau$ matrix multiplications can run at the bottom level. The $4^{N-1}$ comes from the number of parent white blocks in Fig.~\ref{LargeThreadTree} for a depth $N$. The factor of 3 comes from the number of grey blocks in Fig.~\ref{SingleThreadTree}. From the same logic as in Fig.~\ref{BufferFlow}, all threads are constantly applied.

\begin{theorem}
With $\tau=3 \cdot 4^{N-1}$ threads the Partial Strassen algorithm requires $3.25 n^2$ memory.
\label{bettermemory}
\end{theorem}

\begin{proof}
Using $\tau$ threads would result in every level requiring memory for submatrix additions as
\begin{equation}
    2\cdot 3 \cdot 4^{k-1} \frac{n^2}{4^k} = \frac{3}{2}n^2
\end{equation}
where $k \leq N$. In this case, both the A-buffer and B-buffer would be of size $\frac{3}{4}n^2$.

Using $\tau$ threads results in a total memory consumption of $n^2 + \frac{3}{4}n^2 + \frac{3}{2}n^2 = 3.25 n^2$.
\end{proof}

This is opposed to $2.25 n^2$ for the single-threaded implementation or $5.25 n^2$ for the upper bound on a multi-threaded implementation where all operations are performed in parallel. 

Thus, the multi-threaded Partial Strassen variant can have a modest increase in memory consumption, $n^2$ more than the single-threaded version.  Notice that any number of threads could be applied with constant memory scaling, but the specific choice of $\tau$ threads just has special properties for this problem, in particular for the number of operations required.

\subsection{Operations Involved in Multi-Threading}\label{OperationMulti}

If the algorithm is run to depth $N$, then level-$N$ requires $7 \cdot 4^{N-1}$ M-matrices of size $\frac{n^2}{4^{N}}$ to be computed via matrix multiplication, level-$(N-1)$ requires $3 \cdot 4^{N-2}$ M-matrices of size $\frac{n^2}{4^{N-1}}$, level-$(N-2)$ requires $3 \cdot 4^{N-3}$ M-matrices of size $\frac{n^2}{4^{N-2}}$, and so on, until level-$1$ which requires the computation of $3$ M-matrices of size $\frac{n^2}{4}$. With the choice of $\tau = 3 \cdot 4^{N-1}$ threads some particularly nice properties become evident.

The bottom level of multiplications will take $7 \cdot 4^{N-1} / \tau = \frac{7}{3} \approx 2.33$ thread sweeps to compute with one thread allocated to each submatrix multiplication, where the number of thread sweeps is the number of times all threads are applied in parallel to a specific computation such as a matrix operation. For a fractional thread sweep, not all of the threads are applied and so can perform other tasks, but this can be complicated and impractical depending on the hardware used, so we will simply round up to 3 thread sweeps to obtain an upper bound. These submatrices are then combined and stored in $\mathbf{C}$ (requiring only a series of additions). Then there are $3 \cdot 4^{N-2}$ multiplications to be computed, meaning exactly $4$ threads can be applied to each M-matrix of size $\frac{n^2}{4^{N-1}}$. Notice that the matrices involved are $4$ times larger so they can be quartered with the $4$ threads to essentially obtain two times the number of DGEMM multiplications of the same size as was required at the bottom level. At level $k < N$ the number of $\frac{n^2}{4^N}$ matrix multiplications each thread must perform is $2^{N-k}$ (this is the number of sweeps for that level).

We then have that the bottom level requires $3$ thread sweeps, then all above layers require $2^{N-k}$ thread sweeps. Note that the threads compute matrix multiplications of the exact same number of elements at each level for a given sweep (these submatrices can be considered as blocks or slices). Some M-matrices will require more matrix addition operations than others (see Eq.~\eqref{Ms}), but this is a minor perturbation in the computation time for sufficiently large matrices.

\begin{theorem}
With $\tau$ threads the number of multiplication sweeps required for Partial Strassen converges to $75\%$ that of traditional matrix multiplication as depth $N$ increases.
\end{theorem}

\begin{proof}
Compare this multi-threading approach to that for traditional matrix multiplication and level-$N$ Strassen. Assume that each thread performs a matrix multiplication of size $\frac{n^2}{4^{N}}$. Traditional matrix multiplication would require $8^N$ submatrix multiplications of this size, meaning that the total number of thread sweeps would be\footnote{For comparison, level-$N$ Strassen would require $7^N$ submatrix multiplications, such that the number of thread sweeps is $S_{\text{Strassen}}(N) = 7^N/(3 \cdot 4^{N-1})$.
}
\begin{equation}
    S_{\text{traditional}}(N) = \frac{8^N}{3 \cdot 4^{N-1}} = \frac{4}{3}2^N.
\end{equation}
Then, level-$N$ Partial Strassen with this threading scheme, according to our previous discussion, would require sweeps as
\begin{equation}
    S_{\text{partial\;Strassen}}(N) = 3 + \sum_{k=1}^{N-1} 2^k.
\end{equation}

If one takes the ratio between the number of sweeps required for this multi-threaded Partial Strassen algorithm and traditional matrix multiplication, then one finds
\begin{align}
    \frac{S_{\text{partial\;Strassen}}(N)}{S_{\text{traditional}}(N)} &= \frac{3 + \sum_{k=1}^{N-1} 2^k}{\frac{4}{3}2^N} \nonumber \\
    &= \frac{2.25}{2^N} + \frac{3}{4} \sum_{k=1}^{N-1} \left(\frac{1}{2}\right)^{N-k} \nonumber \\
    &= \frac{2.25}{2^N} + \frac{3}{4} \sum_{j=1}^{N-1} \left(\frac{1}{2}\right)^j \nonumber \\
    &= \frac{2.25}{2^N} + \frac{3}{8} \sum_{j=0}^{N-2} \left(\frac{1}{2}\right)^j \nonumber \\
    &= \frac{2.25}{2^N} + \frac{3}{8} \left( \frac{1 - (1/2)^{N-1}}{1-1/2} \right) \nonumber \\
    &= \frac{2.25}{2^N} + \frac{3}{4} \left(1 - \left(\frac{1}{2}\right)^{N-1} \right)
\end{align}
In the limit of large $N$, the above converges to $0.75$.
\end{proof}

This result pairs nicely with the fact that the total number of multiplication operations required for the single-threaded Partial Strassen algorithm converges to $75 \%$ that of traditional matrix multiplication with sufficiently large depth from Thm.~\ref{singlethread_time}. Of course, in both cases there is also overhead incurred from addition operations as well. This limits applicability to sufficiently large matrices, where the boundary for advantage is hardware dependent.

\subsection{Implementation on a GPU}

This paper focuses on CPU implementations (single-threaded and multi-threaded), but we briefly mention here that GPU implementations of Strassen algorithms have been shown to be effective~\cite{Lai, Huang-Yu}. Lai et al. present techniques for implementing an efficient Winograd-Strassen algorithm up to level-4 for double precision operations~\cite{Lai}. Huang et al. presents an efficient scheme for a family of Strassen algorithms which use no more global memory than GEMM by utilizing memory and thread hierarchies~\cite{Huang-Yu}. In theory, similar techniques could be applied to this variant of the Strassen algorithm.

Note that when constructing a GPU implementation of this algorithm, the convention for matrix storage presented in Appendix.~\ref{167method} may be preferable to that of Sec.~\ref{PebblingScheme}.

\section{Results and Discussion}\label{Results}

\begin{figure*}[!t]
    \begin{subfigure}{\columnwidth}
        \caption{}
        \includegraphics[width=\columnwidth]{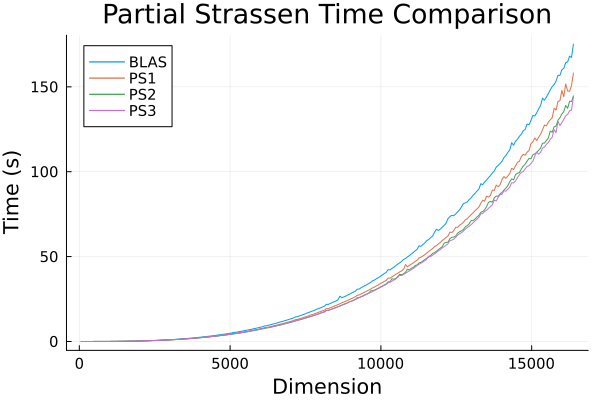}
    \end{subfigure}
    \begin{subfigure}{\columnwidth}
        \caption{}
        \includegraphics[width=\columnwidth]{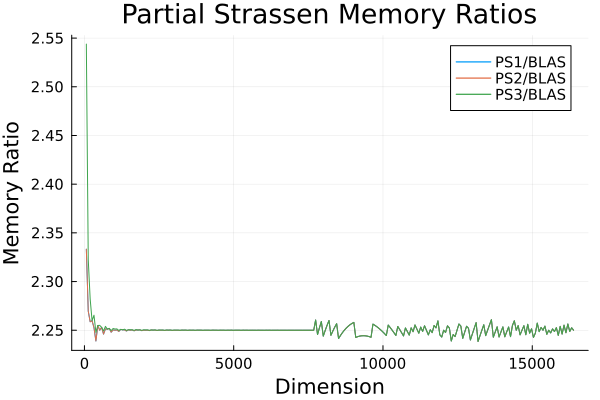}
    \end{subfigure}
    \begin{subfigure}{\columnwidth}
        \caption{}
        \includegraphics[width=\columnwidth]{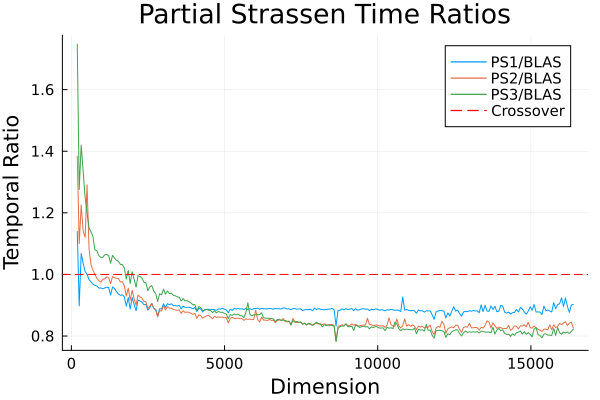}
    \end{subfigure}
    \begin{subfigure}{\columnwidth}
        \caption{}
        \includegraphics[width=\columnwidth]{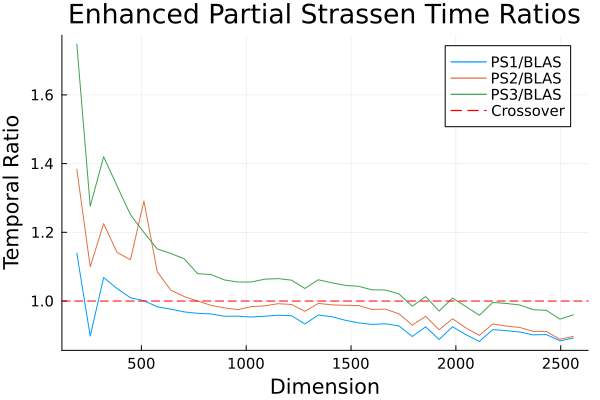}
    \end{subfigure}
    \caption{Experimental data from running the single-level, two-level, and three-level Partial Strassen algorithms, as well as with BLAS matrix multiplication for comparison, on square matrices of equal dimensions. (a) Temporal data for all four algorithms. (b) Memory consumption ratio between each of the three Partial Strassen variants and BLAS. (c) Temporal ratio between each of the three Partial Strassen variants and BLAS. (d) Enhanced version of image (c) depicting the crossover points between each Partial Strassen variant and BLAS.}
    \label{PSData}
\end{figure*}

\begin{figure*}[!t]
    \begin{subfigure}{\columnwidth}
        \caption{}
        \includegraphics[width=\columnwidth]{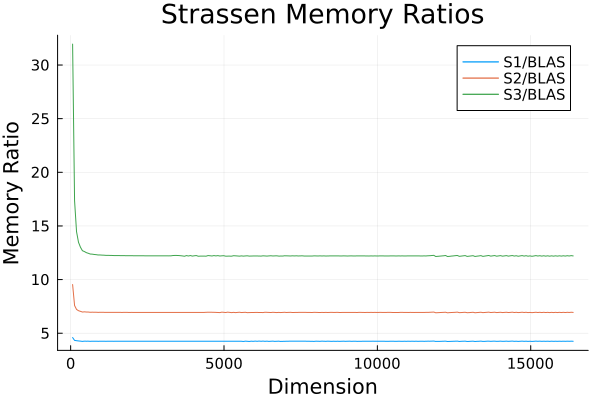}
    \end{subfigure}
    \begin{subfigure}{\columnwidth}
        \caption{}
        \includegraphics[width=\columnwidth]{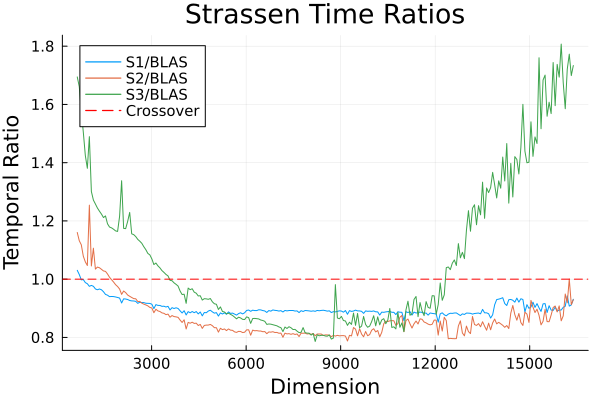}
    \end{subfigure}
    \caption{Experimental data from running the BLAS and conventional Strassen algorithm at depths 1, 2, and 3 on square matrices of equal dimensions. (a) Memory consumption ratio between each of the regular Strassen at three different levels and BLAS. (c) Temporal ratio between each of the regular Strassen algorithm at three different levels and BLAS.}
    \label{StrassenData}
\end{figure*}

\begin{figure*}[!t]
    \begin{subfigure}{\columnwidth}
        \caption{}
        \includegraphics[width=\columnwidth]{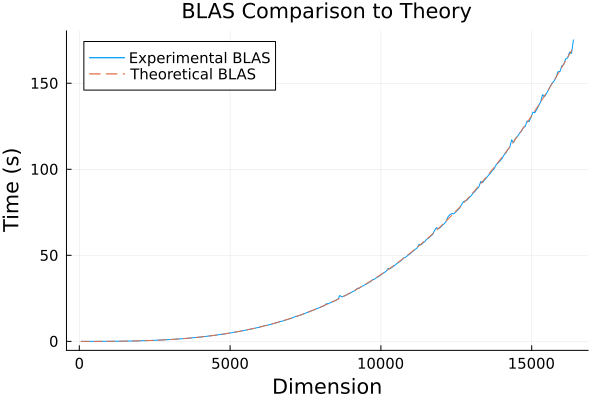}
    \end{subfigure}
    \begin{subfigure}{\columnwidth}
        \caption{}
    \includegraphics[width=\columnwidth]{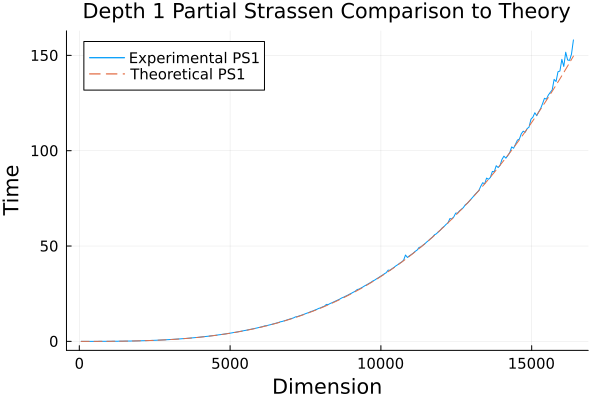}
    \end{subfigure}
    \begin{subfigure}{\columnwidth}
        \caption{}
    \includegraphics[width=\columnwidth]{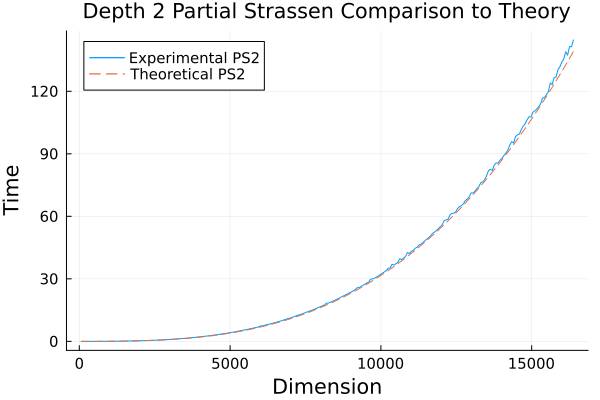}
    \end{subfigure}
    \begin{subfigure}{\columnwidth}
        \caption{}
    \includegraphics[width=\columnwidth]{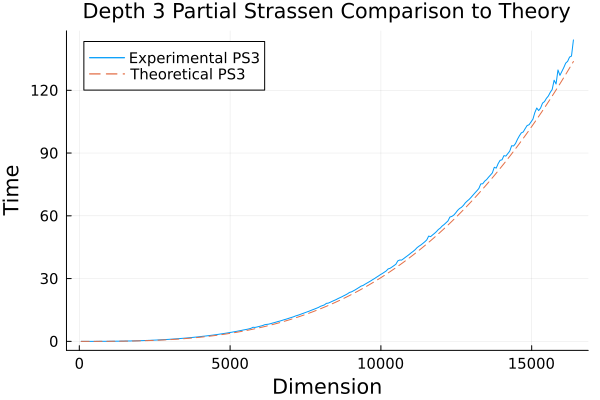}
    \end{subfigure}
    \caption{Theoretical comparisons of (a) BLAS to $n^3$ running time, (b) single-level Partial Strassen algorithm to $0.875 n^3 + 4.5 n^2$, (c) two-level Partial Strassen algorithm to $0.813 n^3 + 11.5 n^2$, three-level Partial Strassen algorithm to $0.781 n^3 + 22.625 n^2$. The scaling factor converting flops to seconds was taken to be $\alpha \approx 3.887 \times 10^{-11} seconds/flop$ found by the linear least squares method to the BLAS data. It should be noted that $O(n)$ subleading terms have been neglected in this comparison.}
    \label{TheoreticalComparison}
\end{figure*}

The Partial Strassen algorithm was run on a single thread at depths 1, 2, and 3 with BLAS matrix multiplication as comparison. The coding language used was Julia 1.9.2. All matrix multiplications were performed for square matrices of equal dimensions from $n=64$ to $16384$ at intervals of $64$ such that $256$ data points were collected. All data was obtained on an Apple M2 processor with 16 GB of RAM. Some additional considerations for implementation are contained in Appendix~\ref{Implementations}.

\subsection{Time and Memory Scaling of the Partial Strassen Algorithm}

The raw times with comparison against BLAS can be seen in Fig.~\ref{PSData}~(a). The more revealing metrics, however, are presented in Figs.~\ref{PSData}~(b, c, d) which show the ratios for metrics of memory consumption and temporal running time in comparison to BLAS.

Fig.~\ref{PSData}~(b) reveals that the amount of memory consumed stays close to the theoretically predicted value of $2.25$ times that of BLAS. Note that the spike for small matrix dimensions in two-level and three-level is due to the fact that padding is required for matrices below a certain size since the number of submatrices involved grows exponentially with depth. The small perturbations in this plot for larger matrix sizes are presumed to be due to fluctuations in the amount of memory consumed either by BLAS or Partial Strassen, notice however that all three memory consumption lines of Partial Strassen completely overlap for sufficiently large matrices.

Fig.~\ref{PSData}~(c) shows that the running time of the three-level Partial Strassen Algorithm converges to $80\%$ that of BLAS (where further improvements may be obtained for even larger matrices), in alignment with theoretical prediction when the $O(n^2)$ factors are taken into account. Fig.~\ref{PSData}~(c) presents crossover points of approximately $500$ for single-level, $750$ for two-level, and $2500$ for three-level (for the specific hardware used). We note that there is significant room for improvement in these crossover points through sophisticated optimization and hardware selection.

\subsection{Comparison with Regular Strassen}

The memory and temporal ratios of the Partial Strassen algorithm depicted in Figs.~\ref{PSData}~(b, c) can be compared with those of the traditional Strassen algorithm in Figs.~\ref{StrassenData}~(a,b). Temporal data from the traditional Strassen algorithm depicts that the two-level Strassen algorithm approaches 80\% the runtime of BLAS for smaller matrix sizes when compared against the three-level Partial Strassen algorithm, however this speedup is short-lived as matrix sizes increase. We allege that this is due to the increased memory consumption of the traditional Strassen algorithm as memory caches become saturated sooner and must begin exchanging data back and forth from RAM. This becomes especially prevalent for the three-level regular Strassen algorithm which consumes nearly 13 times the memory of BLAS and explodes in computation time as matrix sizes increase. The implementation of the regular Strassen algorithm for comparison purposes is available at Ref.~\onlinecite{strassOPen}.

Note that the regular Strassen algorithm has nearly $5n^2$ memory consumption at level one, while the worst-case scaling multi-threaded Partial Strassen algorithm requires just $5.25n^2$ and is depth-independent as from Thm.~\ref{worstcase_threading}. The memory consumption of Thm.~\ref{bettermemory} with $\tau$ threads therefore scales better than the single-level, single-threaded, regular Strassen algorithm as it requires just $3.25n^2$. This places the parallel Partial Strassen algorithm well within the range of utility (with respect to memory consumption concerns).

\subsection{Comparison of Scaling Against Derived Limits}

The theoretical number of operations presented in Table~\ref{AllOps} were used to compare against the observed times. However, a problem arises when one considers the issue of unit mismatch between the theoretical number of flops and the temporal measurements. To remedy this issue and conduct analysis, we obtained an approximate conversion scalar in units of $seconds/flop$ from the BLAS data by assuming the form

\begin{align}\label{scalefactor_expth}
    \vec{t}_\mathrm{exp.} \approx \alpha \vec{t}_\mathrm{th.}
\end{align}
by solving the linear least-squares minimization problem of $|\alpha \vec{t}_\mathrm{th.} - \vec{t}_\mathrm{exp.}|^2$ (which is performed in Julia with the line \lstinline[basicstyle=\footnotesize\ttfamily, mathescape=true]|$\alpha$ = t_th \ t_exp|). We then assumed that this scaling factor was a suitable unit of conversion from flops to seconds and applied it to the theoretical functions for the single-level, two-level, and three-level Partial Strassen algorithms. The theoretical prediction and experimental results are compared in Figs.~\ref{TheoreticalComparison}(a, b, c, d). 

Assuming that the scaling factor from Eq.~\eqref{scalefactor_expth} is accurate, significant agreement can be observed between the theoretical prediction for the Partial Strassen algorithm and the experimental results. However, the three-level variant seems to deviate slightly from prediction, costing more operations than the theoretical minimum potentially due to the absence of $O(n)$ terms incurred by algorithmic overhead (which increases with depth).

\section{Extensions}\label{Extensions}

There may exist other applications for the ideas used in the Partial Strassen algorithm. In particular, there exist other Strassen-like variants which may benefit from similar modification, and the implementation of this bounded memory fractalization scheme may be useful for specialized chip development in low latency applications as it is highly parallelizable.

\subsection{Other Strassen-like Algorithms}\label{StrassenLikeAlgs}

Ever since the Strassen algorithm was first introduced, other variants for matrix multiplication have been discovered which instead of splitting both matrices into quarters ($2 \times 2$ subregions), split them into $n \times k$ and $k \times m$ subregions in what are dubbed Strassen-like algorithms. These types of algorithms are denoted as $\langle n, k, m ; p\rangle$~\cite{Schwartz} where in $A B = C$, $n$ is the outer dimension of subblocks in $A$, $k$ is the inner dimension of subblocks, $m$ is the outer dimension of subblocks in B, and $p$ is the number of subblock multiplications required (other notations include $\langle n, k, m \rangle$ or $(n,k,m)$ each with $p$ being stated separately~\cite{AlphaEvolve}). So traditional matrix multiplication is a $\langle 2, 2, 2; 8 \rangle$, while the single-level Strassen algorithm is a $\langle 2, 2, 2; 7 \rangle$ algorithm. Many such Strassen-like algorithms have been introduced~\cite{AlphaEvolve,Hopcroft,Kauers1,Kauers2,Laderman,Moosbauer,Smirnov1,Smirnov2,Strassen2,Fawzi,Smirnov3}.

The Partial Strassen algorithm for a given depth $N$ can be written in this notation when multiplications of a given level $k$ are weighted accordingly. We can calculate this by simply scaling $8^{N}$ by Eq.~\eqref{MultiplicationFlops} for the corresponding value of $N$ to obtain values for the Partial Strassen algorithm as expressed in Table~\ref{Notations} for depths $N=1$ to $9$.

\begin{table*}[t]
    \centering
        \begin{tabular}{|| l | l | l | l ||}
            \hline
            \textit{\textbf{Depth N}} & \textit{\textbf{Traditional}} &  \textit{\textbf{Strassen}} & \textit{\textbf{Partial Strassen}} \\
            \hline\hline
            $1$ & $\langle 2, 2, 2; 8\rangle$   & $\langle 2, 2, 2; 7\rangle$  & $\langle 2, 2, 2; 7\rangle$\\
            $2$ & $\langle 4, 4, 4; 64\rangle$  & $\langle 4, 4, 4; 49\rangle$ & $\langle 4, 4, 4; 52\rangle$\\
            $3$ & $\langle 8, 8, 8; 512\rangle$ & $\langle 8, 8, 8; 343\rangle$ & $\langle 8, 8, 8; 400\rangle$\\
            $4$ & $\langle 16, 16, 16; 4096\rangle$ & $\langle 16, 16, 16; 2401\rangle$ & $\langle 16, 16, 16; 3136\rangle$\\
            $5$ & $\langle 32, 32, 32; 32768\rangle$ & $\langle 32, 32, 32; 16807\rangle$ & $\langle 32, 32, 32; 24832\rangle$\\
            $6$ & $\langle 64, 64, 64; 262144\rangle$ & $\langle 64, 64, 64; 117649\rangle$ & $\langle 64, 64, 64; 197632\rangle$\\
            $7$ & $\langle 128, 128, 128; 2097152\rangle$ & $\langle 128, 128, 128; 823543\rangle$ & $\langle 128, 128, 128; 1576960\rangle$ \\
            $8$ & $\langle 256, 256, 256; 16777216\rangle$ & $\langle 256, 256, 256; 5764801\rangle$ & $\langle 256, 256, 256; 12599296\rangle$ \\
            $9$ & $\langle 512, 512, 512; 134217728\rangle$ & $\langle 512, 512, 512; 40353607\rangle$ & $\langle 512, 512, 512; 100728832\rangle$ \\
            \hline
        \end{tabular}
    \caption{Comparison of multiplications involved in traditional matrix multiplication for depths 1 through 9, the regular Strassen algorithm, and the Partial Strassen algorithm using $\langle n, k, m ; p\rangle$ notation.}
    \label{Notations}
\end{table*}

We hypothesize that the same techniques used in deriving the Partial Strassen algorithm may be applicable to other Strassen-like algorithms. In particular, there may exist schemes which similarly allow for a subset of submatrix multiplications to be further subdivided by the relevant equations in order for a fixed size buffer to be applied with pebbling game memory fractalization techniques. This would imply highly parallelizable memory-capped Partial Strassen-like algorithms exist. When considering the fact that AlphaEvolve by Google DeepMind \cite{AlphaEvolve} was applied to solve problems in both geometry and derive Strassen-like algorithms, leveraging AI to find other partial variants becomes quite attractive.

\subsection{Strassen Processing Units}

With the historic prevalence of GPU processing and the growing need for rapid matrix multiplication computation in low latency, energy efficient, products (e.g.~autonomous vehicles and robotics) there comes a demand for faster computation. This demand may be satisfied by improving hardware or software. Here, we will briefly consider the proposition and implications of meeting this demand via hardware engineering.

The key advantage of GPUs is that they are general, can be applied to a wide variety of tasks, and have a well developed and documented coding language that allows for massive adoption. Though as lower latency is required, more specialized processing units will be required, units known as application-specific integrated circuits (ASICs). One example of an ASIC is the tensor processing unit (TPU), often used in dedicated AI/ML tasks~\cite{TPU}.

We hypothesize that the memory reduction methods presented in the Partial Strassen algorithm may be applicable to ASIC development of a massively parallel Strassen processing unit for low latency tasks since the required memory is bounded with depth and the algorithm has the potential for parallelizability. This would likely be accomplished by having a set number of tensor cores and a specialized caching structure for optimally efficient computation and memory transfer throughout the chip. These chips could then be used for rapid computation of fixed-size matrices. If Partial Strassen-like algorithms exist as was proposed in Sec.~\ref{StrassenLikeAlgs}, then highly efficient ASICs for a diverse set of matrix shapes may be possible, allowing for greater variety of application. The advantages for creating specialty chips for the Partial Strassen in comparison with its performance on a generic CPU is left for future work and is likely application dependent.

It must be noted that practical/commercial implementation of such a chip may not be viable in the near term due to two main issues. First, current chip architectures seem to fill the needs of the current market. It may not be until more specialized and well-defined applications take shape that this type of architecture becomes a requirement as lower latency becomes vital and energy consumption to commercial success in a competitive landscape. Second, chip development is a lengthy process which requires many iterative cycles in a supply chain where sophisticated chip fabrication is a bottleneck. This means that each iterative cycle would take months, decreasing the probability of near-term application. However, this gives ample reason for near-term investigation.

\section{Conclusion}\label{Conclusion}

We have presented a variant of the Strassen algorithm for matrix multiplication which can theoretically achieve matrix multiplication in $75\%$ the time of traditional implementations with just $2.25$ times the memory when run on a single thread. This algorithm is dubbed the Partial Strassen algorithm as it applies the Strassen equations to a matrix multiplication, then reapplies the Strassen equations to a fraction of submatrix multiplications until the desired depth is reached such that the amount of memory consumed does not change with depth when paired with a pebbling game memory fractalization scheme. The total amount of memory and floating point operations are given up to depth 9, where increased depth is found to give diminishing returns such that limited depth versions of the algorithm are likely to be most applicable. The algorithm is discussed in the context of multi-threading and is similarly shown to theoretically be able to perform matrix multiplication in $75\%$ of the time as traditional matrix multiplication, however this is dependent on the number of threads applied and is less straight-forward than the result for single-threading. An open source version of the algorithm for rectangular matrix multiplication is provided. The algorithm is shown experimentally to perform matrix multiplication in $80\%$ the time of BLAS at depth 3 using $2.25$ the memory on common hardware (Apple M2). Comparison is also made against an implementation of the traditional Strassen algorithm. It is noted similar memory management schemes may be applicable to other Strassen-like variants, and that energy efficient, low-latency, ASIC chips optimized for running the Partial Strassen algorithm in parallel could be produced.

\section{Acknowledgements}

We thank Sarah Huber for useful conversations.

A.D. acknowledges support from the Undergraduate Student Research Award (USRA) from NSERC. A.D. ~acknowledges the NSERC CREATE in Quantum Computing Program, grant number 543245. A.D.~is grateful for the support of the Jamie Cassel's Undergraduate Research Award (JCURA). A.D.~is grateful for support from the Mitacs Accelerate program.

This research was undertaken, in part, thanks to funding from the Canada Research Chairs Program (CRC-2021-00257). This work has been supported in part by the Natural Sciences and Engineering Research Council of Canada (NSERC) under grants RGPIN-2023-05510, DGECR-2023-00026, and ALLRP 590857-23. 

%

\onecolumngrid
\bibliography{refs,TEB_papers}

\begin{appendix}

\section{Partial Strassen Equations}\label{StrassenEquationAppendix}

It can be illustrative to see an example of how the Strassen equations are derived and applied when learning about the subject. Here we do just that, as well as present the full set of $\mathbf{M}_{i}$, $\mathbf{M}_{j,k}$ equations for the two-leveled Partial Strassen algorithm.

\subsection{Multi-leveled Strassen Equation Derivation Example}\label{DeriveStrassenEquation}

To obtain a Strassen equation of a desired depth, denoted $\mathbf{M}_{a, b, c, \ldots, z}$ where $a,b,c,\ldots,z \in \{1,\ldots,7\}$, one must first apply equation $a$ of Eqs.~\eqref{Ms}, then take the multiplication involved and apply equation $b$ of Eqs.~\eqref{Ms} to it, and so on until equation $z$ of Eqs.~\eqref{Ms} is applied. Upon the first application, sub-blocks of input matrices $\mathbf{A}$ and $\mathbf{B}$ are selected. Upon the second application sub-blocks of those sub-blocks are then selected. This process is repeated until the list $a, b, c, \ldots, z$ is exhausted. This can be done algorithmically.

$\mathbf{A}$ and $\mathbf{B}$ submatrices are denoted $\mathbf{X}^{k}_{i,j}$ where $k$ denoting the depth such that the submatrices are size $\frac{n}{2^k} \times \frac{n}{2^k}$, $i$ is the row block, and $j$ is the column block. For example, $\mathbf{A}^2_{3,1}$ is the block matrix of size $\frac{n}{4} \times \frac{n}{4}$ with starting row index $\frac{2n}{4}+1$ and starting column index $1$. 

Suppose you apply the Strassen equations to a matrix multiplication and wish apply the Strassen equations once again to the multiplication operation involved in $\mathbf{M}_2$ of Eqs.~\eqref{Ms}. The matrix $\mathbf{M}_{2,1}$ may be obtained as follows.

Take $\mathbf{M}_2 = (\mathbf{A}^1_{2,1} + \mathbf{A}^1_{2,2}) \cdot \mathbf{B}^1_{1,1}$ and apply equation $\mathbf{M}_1 = (\mathbf{A}^1_{1,1} + \mathbf{A}^1_{2,2})\cdot(\mathbf{B}^1_{1,1} + \mathbf{B}^1_{2,2})$ to subdivide the blocks of $\mathbf{A}$ and $\mathbf{B}$. For $\mathbf{A}$, submatrices $(1, 1)$ and $(2, 2)$ of $\mathbf{A}^1_{2,1}$ are $\mathbf{A}^2_{3,1}$ and $\mathbf{A}^2_{4,2}$; those of $\mathbf{A}^1_{2,2}$ are $\mathbf{A}^2_{3,3}$ and $\mathbf{A}^2_{4,4}$. For $\mathbf{B}$, submatrices $(1, 1)$ and $(2, 2)$ of $\mathbf{B}^1_{1,1}$ are $\mathbf{B}^2_{1,1}$ and $\mathbf{B}^2_{2,2}$. The resulting equation is then 

\begin{equation}
    \mathbf{M}_{2,1} = \left(\mathbf{A}^2_{3,1} + \mathbf{A}^2_{4,2} + \mathbf{A}^2_{3,3} + \mathbf{A}^2_{4,4}\right) \cdot \left(\mathbf{B}^2_{1,1} + \mathbf{B}^2_{2,2}\right)
\end{equation}
This procedure is represented pictorially in Fig.~\ref{M21}.\\

\begin{figure}[H]
    \centering
    \includegraphics[width=0.9\columnwidth]{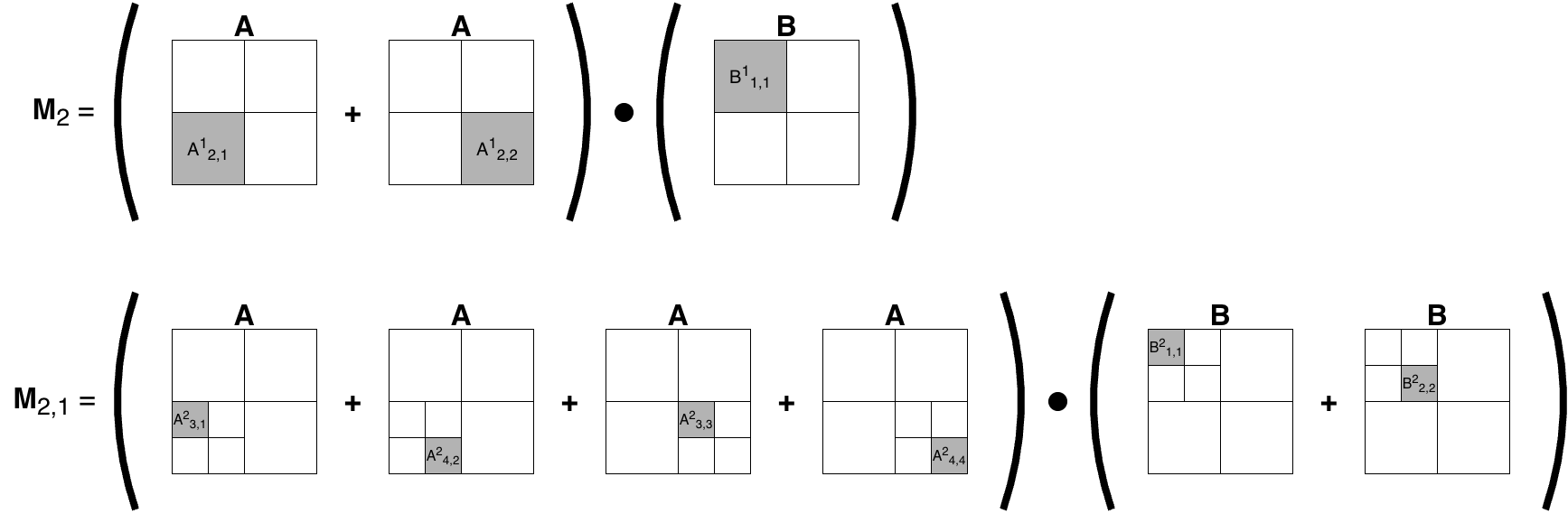}
    \caption{Diagrammatic representation of $\mathbf{M}_2$ (top) and $\mathbf{M}_{2,1}$ (bottom) which are computed by adding and multiplying submatrices of $\mathbf{A}$ and $\mathbf{B}$. Grey blocks indicate the blocks to be added then multiplied. In representation for $\mathbf{M}_2$ the grey blocks are of size $\frac{n}{2} \times \frac{n}{2}$, indicating that the matrix addition and multiplication operations occur between matrices of that size. This is similarly the case for $\mathbf{M}_{2,1}$, where block sizes are instead $\frac{n}{4} \times \frac{n}{4}$.}
    \label{M21}
\end{figure}

\subsection{Strassen Equation Propagation}\label{EquationPropagation}

Here we give an example of how an $\mathbf{M}$-matrix of a depth $N=2$ Strassen algorithm propagates up to the top level after instantiation as described in Appendix~\ref{DeriveStrassenEquation}. Take $\mathbf{M}_{2,1}$ to be the matrix of consideration. It must be combined with the other 6 $\mathbf{M}_{2,i}$ matrices   where $i \in \{2,3,\ldots,7\}$ according to Eq.~\eqref{C} to form $\mathbf{M}_2$ as
%
\begin{equation}
    \mathbf{M}_2 = \left(\begin{array}{cc}
        \mathbf{M}_{2,4} + \textcolor{blue}{\mathbf{M}_{2,1}} - \mathbf{M}_{2,5} + \mathbf{M}_{2,7} & \mathbf{M}_{2,5} + \mathbf{M}_{2,3}\\
        \mathbf{M}_{2,2} + \mathbf{M}_{2,4} & \mathbf{M}_{2,6} + \textcolor{blue}{\mathbf{M}_{2,1}} - \mathbf{M}_{2,2} + \mathbf{M}_{2,3}
    \end{array}\right)
\end{equation}
then $\mathbf{M}_2$ is combined with the other 6 $\mathbf{M}_j$ matrices where $j \in \{1,3,4,\ldots,7\}$

\begin{equation}
    \mathbf{C} = \left(\begin{array}{cc}
        \mathbf{M}_4 + \mathbf{M}_1 - \mathbf{M}_5 + \mathbf{M}_7 & \mathbf{M}_5 + \mathbf{M}_3\\
        \textcolor{blue}{\mathbf{M}_2} + \mathbf{M}_4                               & \mathbf{M}_6 + \mathbf{M}_1 - \textcolor{blue}{\mathbf{M}_2} + \mathbf{M}_3
    \end{array}\right)
\end{equation}
yielding $\mathbf{C}$ as the resultant matrix.

\subsection{The Two-Level Partial Strassen Equations}\label{TwoLevelTableAppendix}

The equations to compute all $\mathbf{M}_i$ and $\mathbf{M}_{j,k}$ matrices for a two-level Partial Strassen algorithm are presented in Table.~\ref{TwoLevel}.

\begin{table}[H]
    \begin{center}
        \begin{tabular}{|| l | l ||}
            \hline
            \textit{\textbf{Matrix}} & \textit{\textbf{Equation}} \\
            \hline\hline
            $\mathbf{M}_1$ & $(\mathbf{A}^1_{1,1} + \mathbf{A}^1_{2,2}) \cdot (\mathbf{B}^1_{1,1} + \mathbf{B}^1_{2,2})$ \\

            \hline
            $\mathbf{M}_{2,1}$ & $(\mathbf{A}^2_{3,1} + \mathbf{A}^2_{3,3} + \mathbf{A}^2_{4,2} + \mathbf{A}^2_{4,4}) \cdot (\mathbf{B}^2_{1,1} + \mathbf{B}^2_{2,2})$ \\

            $\mathbf{M}_{2,2}$ & $(\mathbf{A}^2_{4, 1} + \mathbf{A}^2_{4, 2} + \mathbf{A}^2_{4, 3} + \mathbf{A}^2_{4, 4}) \cdot (\mathbf{B}^2_{1, 1})$ \\

            $\mathbf{M}_{2,3}$ & $(\mathbf{A}^2_{3, 1} + \mathbf{A}^2_{3, 3}) \cdot (\mathbf{B}^2_{1, 2} - \mathbf{B}^2_{2, 2})$ \\

            $\mathbf{M}_{2,4}$ & $(\mathbf{A}^2_{4, 2} + \mathbf{A}^2_{4, 4}) \cdot (\mathbf{B}^2_{2, 1} - \mathbf{B}^2_{1, 1})$ \\

            $\mathbf{M}_{2,5}$ & $(\mathbf{A}^2_{3, 1} + \mathbf{A}^2_{3, 2} + \mathbf{A}^2_{3, 3} + \mathbf{A}^2_{3, 4}) \cdot (\mathbf{B}^2_{2, 2})$ \\

            $\mathbf{M}_{2,6}$ & $(\mathbf{A}^2_{4, 1} - \mathbf{A}^2_{3, 1} + \mathbf{A}^2_{4, 3} - \mathbf{A}^2_{3, 3}) \cdot (\mathbf{B}^2_{1, 1} + \mathbf{B}^2_{1, 2})$ \\

            $\mathbf{M}_{2,7}$ & $(\mathbf{A}^2_{3, 2} - \mathbf{A}^2_{4, 2} + \mathbf{A}^2_{3, 4} - \mathbf{A}^2_{4, 4}) \cdot (\mathbf{B}^2_{2, 1} + \mathbf{B}^2_{2, 2})$ \\

            \hline
            $\mathbf{M}_3$ & $(\mathbf{A}^1_{1, 1}) \cdot (\mathbf{B}^1_{1, 2} - \mathbf{B}^1_{2, 2})$ \\

            \hline
            $\mathbf{M}_{4,1}$ & $(\mathbf{A}^2_{3, 3} + \mathbf{A}^2_{4, 4}) \cdot (\mathbf{B}^2_{3, 1} + \mathbf{B}^2_{4, 2} - \mathbf{B}^2_{1, 1} - \mathbf{B}^2_{2, 2})$ \\

            $\mathbf{M}_{4,2}$ & $(\mathbf{A}^2_{4, 3} + \mathbf{A}^2_{4, 4}) \cdot (\mathbf{B}^2_{3, 1} - \mathbf{B}^2_{1, 1})$ \\

            $\mathbf{M}_{4,3}$ & $(\mathbf{A}^2_{3, 3}) \cdot (\mathbf{B}^2_{3, 2} - \mathbf{B}^2_{4, 2} - \mathbf{B}^2_{1, 2} + \mathbf{B}^2_{2, 2})$ \\

            $\mathbf{M}_{4,4}$ & $(\mathbf{A}^2_{4, 4}) \cdot (\mathbf{B}^2_{4, 1} - \mathbf{B}^2_{3, 1} - \mathbf{B}^2_{2, 1} + \mathbf{B}^2_{1, 1})$ \\

            $\mathbf{M}_{4,5}$ & $(\mathbf{A}^2_{3, 3} + \mathbf{A}^2_{3, 4}) \cdot (\mathbf{B}^2_{4, 2} - \mathbf{B}^2_{2, 2})$ \\

            $\mathbf{M}_{4,6}$ & $(\mathbf{A}^2_{4, 3} - \mathbf{A}^2_{3, 3}) \cdot (\mathbf{B}^2_{3, 1} + \mathbf{B}^2_{3, 2} - \mathbf{B}^2_{1, 1} - \mathbf{B}^2_{1, 2})$ \\

            $\mathbf{M}_{4,7}$ & $(\mathbf{A}^2_{3, 4} - \mathbf{A}^2_{4, 4}) \cdot (\mathbf{B}^2_{4, 1} + \mathbf{B}^2_{4, 2} - \mathbf{B}^2_{2, 1} - \mathbf{B}^2_{2, 2})$ \\

            \hline
            $\mathbf{M}_{5,1}$ & $(\mathbf{A}^2_{1, 1} + \mathbf{A}^2_{2, 2} + \mathbf{A}^2_{1, 3} + \mathbf{A}^2_{2, 4}) \cdot (\mathbf{B}^2_{3, 3} + \mathbf{B}^2_{4, 4})$ \\

            $\mathbf{M}_{5,2}$ & $(\mathbf{A}^2_{2, 1} + \mathbf{A}^2_{2, 2} + \mathbf{A}^2_{2, 3} + \mathbf{A}^2_{2, 4}) \cdot (\mathbf{B}^2_{3, 3})$ \\

            $\mathbf{M}_{5,3}$ & $(\mathbf{A}^2_{1, 1} + \mathbf{A}^2_{1, 3}) \cdot (\mathbf{B}^2_{3, 4} - \mathbf{B}^2_{4, 4})$ \\

            $\mathbf{M}_{5,4}$ & $(\mathbf{A}^2_{2, 2} + \mathbf{A}^2_{2, 4}) \cdot (\mathbf{B}^2_{4, 3} - \mathbf{B}^2_{3, 3})$ \\

            $\mathbf{M}_{5,5}$ & $(\mathbf{A}^2_{1, 1} + \mathbf{A}^2_{1, 2} + \mathbf{A}^2_{1, 3} + \mathbf{A}^2_{1, 4}) \cdot (\mathbf{B}^2_{4, 4})$ \\

            $\mathbf{M}_{5,6}$ & $(\mathbf{A}^2_{2, 1} - \mathbf{A}^2_{1, 1} + \mathbf{A}^2_{2, 3} - \mathbf{A}^2_{1, 3}) \cdot (\mathbf{B}^2_{3, 3} + \mathbf{B}^2_{3, 4})$ \\

            $\mathbf{M}_{5,7}$ & $(\mathbf{A}^2_{1, 2} - \mathbf{A}^2_{2, 2} + \mathbf{A}^2_{1, 4} - \mathbf{A}^2_{2, 4}) \cdot (\mathbf{B}^2_{4, 3} + \mathbf{B}^2_{4, 4})$ \\

            \hline
            $\mathbf{M}_{6,1}$ & $(\mathbf{A}^2_{3, 1} + \mathbf{A}^2_{4, 2} - \mathbf{A}^2_{1, 1} - \mathbf{A}^2_{2, 2}) \cdot (\mathbf{B}^2_{1, 1} + \mathbf{B}^2_{2, 2} + \mathbf{B}^2_{1, 3} + \mathbf{B}^2_{2, 4})$ \\

            $\mathbf{M}_{6,2}$ & $(\mathbf{A}^2_{4, 1} + \mathbf{A}^2_{4, 2} - \mathbf{A}^2_{2, 1} - \mathbf{A}^2_{2, 2}) \cdot (\mathbf{B}^2_{1, 1} + \mathbf{B}^2_{1, 3})$ \\

            $\mathbf{M}_{6,3}$ & $(\mathbf{A}^2_{3, 1} - \mathbf{A}^2_{1, 1}) \cdot (\mathbf{B}^2_{1, 2} - \mathbf{B}^2_{2, 2} + \mathbf{B}^2_{1, 4} - \mathbf{B}^2_{2, 4})$ \\

            $\mathbf{M}_{6,4}$ & $(\mathbf{A}^2_{4, 2} - \mathbf{A}^2_{2, 2}) \cdot (\mathbf{B}^2_{2, 1} - \mathbf{B}^2_{1, 1} + \mathbf{B}^2_{2, 3} - \mathbf{B}^2_{1, 3})$ \\

            $\mathbf{M}_{6,5}$ & $(\mathbf{A}^2_{3, 1} + \mathbf{A}^2_{3, 2} - \mathbf{A}^2_{1, 1} - \mathbf{A}^2_{1, 2}) \cdot (\mathbf{B}^2_{2, 2} + \mathbf{B}^2_{2, 4})$ \\

            $\mathbf{M}_{6,6}$ & $(\mathbf{A}^2_{4, 1} - \mathbf{A}^2_{3, 1} - \mathbf{A}^2_{2, 1} + \mathbf{A}^2_{1, 1}) \cdot (\mathbf{B}^2_{1, 1} + \mathbf{B}^2_{1, 2} + \mathbf{B}^2_{1, 3} + \mathbf{B}^2_{1, 4})$ \\

            $\mathbf{M}_{6,7}$ & $(\mathbf{A}^2_{3, 2} - \mathbf{A}^2_{4, 2} - \mathbf{A}^2_{1, 2} + \mathbf{A}^2_{2, 2}) \cdot (\mathbf{B}^2_{2, 1} + \mathbf{B}^2_{2, 2} + \mathbf{B}^2_{2, 3} + \mathbf{B}^2_{2, 4})$ \\

            \hline
            $\mathbf{M}_7$ & $(\mathbf{A}^1_{1,2} - \mathbf{A}^1_{2,2}) \cdot (\mathbf{B}^1_{2,1} + \mathbf{B}^1_{2,2})$  \\ 

            \hline
        \end{tabular}
    \end{center}
    \caption{Equations for the two-level Partial Strassen algorithm.}
    \label{TwoLevel}
\end{table}

\section{Implementations}\label{Implementations}

During the development of this algorithm, many versions were implemented and compared against one another. All versions created will be open sourced in the StrassOPen.jl GitHub repository~\cite{strassOPen}.

\subsection{Alternate Storage Method}\label{167method}

In this paper, we discussed a version of the algorithm which stored the produced $\mathbf{M}_i$ matrices as outlined in Eq.~\eqref{Cstore} and Eq.~\eqref{Dstore}. There is, however, at least one other option where submatrices are stored as 

\begin{equation}\label{Cstore-variant}
    \mathbf{C} \leftarrow \left(\begin{array}{cc}
        \mathbf{M}_4 & \mathbf{M}_5 \\
        \mathbf{M}_2 & \mathbf{M}_3
    \end{array}\right)
\end{equation}
and
\begin{equation}\label{Dstore-variant}
    \mathbf{\mathcal{M}} \leftarrow \left(\begin{array}{cc}
        \mathbf{M}_1 & \mathbf{M}_6 \\
        \mathbf{M}_7 & 0
    \end{array}\right)
\end{equation}
where $\mathbf{M}_1$, $\mathbf{M}_6$, and $\mathbf{M}_7$ stored in $\mathbf{\mathcal{M}}$ instead of $\mathbf{M}_1$, $\mathbf{M}_3$, $\mathbf{M}_7$.

This choice has an upside and a downside. Notice that from Eqs.~\eqref{Ms} that only the $\mathbf{M}_i$ matrices which require a single addition operation are split. This reduces the total number of additions involved. However, this choice also required another buffer to be used. This can be done with a single register on each sequential addition operation in the recombination step.

Due to this addition structure, this alternative storage method may be advantageous for GPU implementations since the addition operations performed in the computation of each $\mathbf{M}_{\ldots,i}$ at each level can be standardized such that each thread can follow the same operation set rather than having branching conditional statements (which can create significant slowdowns on some GPU hardware).

The added complexity added by this choice and the fact that Strassen variants usually run to limited depth led us to favor the storage method of Eq.~\eqref{Cstore} and Eq.~\eqref{Dstore} for the single threaded implementation. Both storage methods were implemented for a single-threaded version of the algorithm, and no significant difference was found during preliminary tests.

\subsection{Tree method vs. List method}

In line with Fig.~\ref{LargeThreadTree}, one proposed implementation was to create a tree data structure with metadata about which operation to perform at what point in the algorithm. In theory, this design is algorithmically pleasing. Though the increased overhead incurred by accessing various data types and the use of mutex locks to prevent race conditions within the tree structure during multi-threading caused this approach to prove impractical when compared against simpler implementations.

In the end, a list method was chosen such that metadata could be accessed efficiently. In this approach, levels of the tree can be split into multiple lists for multi-threading application. This method proved to incur slightly less overhead in terms of computation time and memory usage.

\subsection{Implementation Details}

As mentioned, the final implementation of this algorithm was implemented with storage according to Eq.~\eqref{Cstore} and Eq.~\eqref{Dstore}, and with operation metadata stored in list format to reduce overhead. Equations are obtained algorithmically such that they need only be computed once for a desired depth. Then dictionaries can be initialized for a given pair of input matrices such that metadata can be used to access the range of indices to operate over for a given matrix computation.

The use of dictionaries reduces the number of computations which must be performed during execution, but is not strictly necessary. The algorithm would perform similarly without the use of dictionaries, but the crossover point may increase due to increased overhead.

For multi-threading, simply applying BLAS with a different number of threads at each level would create additional overhead. So instead we implemented all additions and matrix multiplications in a single-threaded manner. When multiple threads are to be applied to a matrix computation of Eqs.~\eqref{Ms}, we chose to split the matrix into vertical and horizontal slices rather than subblocks. For computations according to Eqs.~\eqref{C}, we chose to add subblocks, but this could be changed trivially.

\end{appendix}

\end{document}